%% file: ms.tex
\documentclass[a4paper,UKenglish]{lipics}
\usepackage{amsthm}
\theoremstyle{plain}
\newtheorem{mylemma}{Lemma}

\usepackage{booktabs}
\newtheorem{introthm}{Theorem}

\newcommand{\Patrascu}{P\v{a}tra\c{s}cu}
\newcommand{\modulo}{\operatorname{mod}{}}

\newcommand{\remove}[1]{}

\newcommand{\CLB}{CLB}
\newcommand{\ThreeSUM}{\textsf{3SUM}}

\newcommand{\ConvolutionThreeSUM}{\textsf{Convolution\ThreeSUM}}

\newcommand{\DiffConvolutionThreeSUM}{\textsf{DiffConv\ThreeSUM}}

\newtheorem{observation}{Observation}

\begin{document}

\title{How Hard is it to Find (Honest) Witnesses?}
\date{}

\author[1]{Isaac Goldstein\thanks{This research is supported by the Adams Foundation of the Israel Academy of Sciences and Humanities}}
\author[2]{Tsvi Kopelowitz\thanks{This research is supported by NSF grants CCF-1217338, CNS-1318294, and CCF-1514383}}
\author[1]{Moshe Lewenstein\thanks{This research is supported by a BSF grant 2010437 and a GIF grant 1147/2011.}}
\author[1]{Ely Porat}
\affil[1]{Bar-Ilan University \\  \texttt{\{goldshi,moshe,porately\}@cs.biu.ac.il}}
\affil[2]{University of Michigan \\  \texttt{kopelot@gmail.com}}
\authorrunning{I. Goldstein, T. Kopelowitz, M. Lewenstein and E. Porat}
\renewcommand\Authands{ and }

\maketitle

\thispagestyle{empty}

\input{abstract}

\input{introduction}

\setcounter{theorem}{0}

\input{conv_witnesses}

\input{partial_conv}

\input{partial_matrix_mult}

\input{jumbled_indexing}

\bibliographystyle{plain} \bibliography{ms}

\newpage
\appendix
\section*{Appendix}

\input{summary_of_applications}

\newpage
\input{appendix}

\end{document}

%% file: abstract.tex
\begin{abstract}
In recent years much effort was put into developing polynomial-time conditional lower bounds for algorithms and data structures in both static and dynamic settings. Along these lines we suggest a framework for proving conditional lower bounds based on the well-known 3SUM conjecture. Our framework creates a \emph{compact representation} of an instance of the 3SUM problem using hashing and domain specific encoding. This compact representation admits false solutions to the original 3SUM problem instance which we reveal and eliminate until we find a true solution. In other words, from all \emph{witnesses} (candidate solutions) we figure out if an \emph{honest} one (a true solution) exists. This enumeration of witnesses is used to prove conditional lower bound on \emph{reporting} problems that generate all witnesses. In turn, these reporting problems are reduced to various decision problems. These help to enumerate the witnesses by constructing appropriate search data structures. Hence, 3SUM-hardness of the decision problems is deduced.

We utilize this framework to show conditional lower bounds for several variants of convolutions, matrix multiplication and string problems. Our framework uses a strong connection between all of these problems and the ability to find \emph{witnesses}.

Specifically, we prove conditional lower bounds for computing partial outputs of convolutions and matrix multiplication for sparse inputs. These problems are inspired by the open question raised by Muthukrishnan 20 years ago \cite{Muthukrishnan95}. The lower bounds we show rule out the possibility (unless the 3SUM conjecture is false) that almost linear time solutions to sparse input-output convolutions or matrix multiplications exist. This is in contrast to standard convolutions and matrix multiplications that have, or assumed to have, almost linear solutions.

Moreover, we improve upon the conditional lower bounds of Amir et al.~\cite{ACLL14} for histogram indexing, a problem that has been of much interest recently. The conditional lower bounds we show apply for both reporting and decision variants. For the well-studied decision variant, we show a full tradeoff between preprocessing and query time for every alphabet size > 2. At an extreme, this implies that no solution to this problem exists with subquadratic preprocessing time and $\tilde{O}(1)$ query time for every alphabet size > 2, unless the 3SUM conjecture is false. This is in contrast to a recent result by Chan and Lewenstein~\cite{CL15} for a binary alphabet.

While these specific applications are used to demonstrate the techniques of our framework, we believe that this novel framework is useful for many other problems as well.

\end{abstract}

%% file: introduction.tex
\section{Introduction}

In recent years much effort has been invested towards developing polynomial time lower bounds for algorithms and data structures in both static and dynamic settings. This effort is directed towards obtaining a better understanding of the complexity class $P$ for well-studied problems which seem hard in the polynomial sense. The seminal paper by Gajentaan and Overmars \cite{GO95} set the stage for this approach by proving lower bounds for many problems in computational geometry conditioned on the \ThreeSUM{} conjecture. In the \ThreeSUM{} problem we are given a set $A$ of $n$ integers and we need to establish if there are $a,b,c \in A$ such that $a+b+c=0$. This problem has a simple $O(n^2)$ algorithm (and some poly-logarithmic improvements in ~\cite{BDP05,GronlundP14}) but no truly subquadratic algorithm is known, where truly subquadratic means $O(n^{2-\epsilon})$ for some $\epsilon>0$. The \ThreeSUM{} conjecture states that no truly subquadratic algorithm exists for the \ThreeSUM{} problem.
Based on this conjecture, there has been a recent extensive line of work establishing conditional lower bounds (\CLB{}s) for many problems in a variety of fields other than computational geometry, including many interesting dynamic problems, see e.g. \cite{AL13,AW14,AWW14,AWY15, KPP16,Patrascu10}.

\subsection{Decision and Reporting Problems}
Algorithmic problems come in many flavors. The classic one is the \emph{decision} variant. In this variant, we are given an instance of a problem and we are required to decide if it has some property or not. Some examples include: (1) given a 3-CNF formula we may be interested in deciding if it satisfiable by some truth assignment; (2) given a bipartite graph we may be interested in deciding if the graph has a perfect matching; (3) given a text $T$ and a pattern $P$ we may be interested in deciding if $P$ occurs in $T$. It is well-known that the first example is NP-complete while the two others are in P.
An instance that has the property in question has at least one \emph{witness} that proves the existence of the property. In the examples above a witness is: (1) a satisfying assignment; (2)  a perfect matching in the graph; (3) a position of an occurrence of $P$ in $T$.
Sometimes, we are not only interested in understanding if a witness \emph{exists}, but rather we wish to \emph{enumerate} all of the witnesses. This is the \emph{reporting} variant of the problem. In the examples mentioned above the goal of the reporting variant is to: (1) enumerate all satisfying assignments; (2) enumerate all perfect matchings; (3) enumerate all occurrences of $P$ in $T$.
For the first two examples it is known from complexity theory that it is most likely hard to count the number of witnesses (not to mention reporting them) (these are \#P-complete problems), while the third example can be solved by classic linear time algorithms.

In this paper we investigate the interplay between the decision and reporting variants of algorithmic problems and present a systematic framework that is used for proving \CLB{}s for these variants. We expect this framework to be useful for proving \CLB{}s on other problems not considered here.

\subsection{Our Framework}

We introduce and follow a framework that shows \ThreeSUM{}-hardness of decision problems via their reporting versions. The high-level idea is to reduce an instance of \ThreeSUM{} to an instance of a reporting problem, and then reduce the instance of a reporting problem to several instances of its decision version using a sophisticated search structure. The outline of this framework is described next.

\begin{itemize}
\item{\textbf{Compact Representation}. One of the difficulties in proving \CLB{}s based on the \ThreeSUM{} conjecture is that the input universe for \ThreeSUM{} could be too large for accommodating a reduction to a certain problem.
    To tackle this, we \emph{embed} the universe using special hashing techniques. This is sometimes coupled with a secondary problem-specific encoding scheme in order to match the problem at hand. }
\item{\textbf{Reporting}. The embedding in the first step may introduce false-positives. To tackle this, we report {\em all} the candidate solutions (witnesses) for the embedded \ThreeSUM{} instance, in order to verify if a true solution (an honest witness) to \ThreeSUM{} really exists.
    This is where we are able to say something about the difficulty of solving reporting problems. This is done by reducing the embedded \ThreeSUM{} instance to an instance of such a reporting problem, if it provides an efficient way to find all the false-positives. In some cases, such reductions reveal tradeoff relationships between the preprocessing time and reporting/query time.}
\item{\textbf{Reporting via Decision}. In this step the goal is to establish \ThreeSUM-hardness of a decision problem. To do so we reduce an instance of the reporting version of the problem to instances of the decision version by creating a data structure on top of the many instances of the decision version. This data structure allows us to efficiently report all of the elements in the output of the instance of the reporting version.
    By constructing the data structure in different ways we obtain varying \CLB{s} for the decision variants depending on the specific structure that we use.}
\end{itemize}

By following this route we introduce new \CLB{s} for some important problems which are discussed in detail in Section~\ref{sec:applications}. We point out that the embedding in the first step follows along the lines of \cite{Patrascu10} and~\cite{KPP16}. However, in some cases we also add an additional encoding scheme to fit the needs of the specific problem at hand.

\medskip

\noindent\textbf{Implications.} In Section~\ref{sec:applications} we discuss three applications from two different domains which utilize our framework for proving \CLB{s}, thereby demonstrating the usefulness of our framework. Table~\ref{table:summary_of_applications} in
Appendix~\ref{sec:summary_of_applications} summarizes these results. Of particular interest are new results on Histogram Indexing (defined in Section~\ref{sec:applications}) which, together with the algorithm of~\cite{CL15}, demonstrate a sharp separation when allowing truly subquadratic preprocessing time between binary and trinary alphabet settings. Moreover, our framework is the first to obtain a \CLB{} for the reporting version, which, as opposed to the decision variant, also holds for the binary alphabet case. See Table~\ref{table:histogram_indexing_comparison} in Appendix \ref{sec:summary_of_applications}.

\section{Applications}\label{sec:applications}

\noindent
 {\bf Convolution Problems\ }

 \noindent
 The \emph{convolution} of two vectors $u,v\in \{\mathbb{R}^+\cup\{0\}\}^n$ is a vector $w$, such that $w[k]=\sum_{i=0}^{k}{u[i]v[k-i]}$ for $0 \leq k \leq 2n-2$. Computing the convolution of $u$ and $v$ takes  $O(n\log{n})$ time using the celebrated FFT algorithm. Convolutions are used extensively in many areas including signal processing, communications, image compression, pattern matching, etc.
  A \emph{convolution witness} for the $k$th entry in $w$ is a pair $(a,b)$ such that $a+b=k$ and $u[a]\cdot v[b] > 0$. In other words, the witnesses of entry $k$ in $w$ are all values $i$ that contribute a non-zero value to the summation $w[k] = \sum_{i=0}^{k}{u[i]v[k-i]}$.
The first convolution problem we consider is the {\em convolution witnesses problem} which is defined as follows.

\begin{definition}
In the {\bf convolution witnesses problem} we preprocess two vectors $u,v\in \{\mathbb{R}^+\cup\{0\}\}^n$ and their convolution vector $w$, so that given a query integer $0\leq k \leq 2n-2$, we list {\em all} convolution witnesses of index $k$ in $w$.
\end{definition}

We prove the following \CLB{} for the convolution witnesses problem that holds even if $u$ and $v$ are binary vectors and all numbers in $w$ are non-negative integers.

\begin{introthm}~\label{thm:convolution-witnesses}
Assume the \ThreeSUM{} conjecture is true.
Then for any constant $0<\alpha <1$, there is no algorithm solving the convolution witnesses problem with $O(n^{2-\alpha})$ expected preprocessing time and $O(n^{\alpha/2 - \Omega(1)})$ expected amortized query time {\bf per witness}.
\end{introthm}

Theorem~\ref{thm:convolution-witnesses} implies that when using only truly subquadratic preprocessing time one is required to spend a significant polynomial amount of time on every single witness. In particular, this means that, assuming the \ThreeSUM{} conjecture, one cannot expect to find witnesses much faster than following the naive algorithm for computing convolution na\"{\i}vely according to the convolution definition. This is in contrast to the decision version of the problem, where we only ask if a witness exists. This variant is easily solved using constant query time after a near linear time preprocessing procedure (computing the convolution itself).

Another variation of the convolution problem which we consider is the \emph{sparse convolution problem}. There are two different problems named sparse convolution, both appearing as open questions in a paper by Muthukrishnan \cite{Muthukrishnan95}. In the first, which is now well understood, we are given Boolean vectors $u$ and $v$ of lengths $N$ and $M$, where $M<N$. There are $n$ ones in $u$, $m$ ones in $v$ and $z$ ones in $w$, where $w$ is the Boolean convolution vector of $u$ and $v$. The goal is to report the non-zero elements in $w$ in $\tilde{O}(z)$ time. This problem has been extensively studied, and the goal has been achieved; see for example \cite{CL15,CH02,HIKP12}.
The second variant which we call {\em partial convolutions} is as follows.

\begin{definition} The {\bf partial convolution problem} on two vectors $u$ and $v$ of real numbers (of length $N$ and $M$ respectively, where $M<N$) and a set $S$ of indices is to compute, for each $i \in S$, the value of the $i$-th element in the convolution of $u$ and $v$.
\end{definition}

Muthukrishnan in~\cite{Muthukrishnan95} asked if it is possible to compute a partial convolution significantly faster than the time needed to compute a (classic) convolution.  We prove a \CLB{} based on the \ThreeSUM{} conjecture, that holds also for the special case of Boolean vectors, and, therefore, also for the special case in which we only want to know if the output values at indices in $S$ are zero or more. Moreover, we focus on the important
%\tknote{why is this important?  am I missing something?}
variant of this problem that deals with the case where the two input vectors have only $n=O(N^{1-\Omega(1)})$ ones and are both given implicitly (specifying only the indices of the ones). Our results also extend to the indexing version of the partial convolution problem, which we call the {\em partial convolution indexing problem}, and is defined as follows.% (notice that here we extend the definition to include more general universes).

\begin{definition}  The \emph{\bf partial convolution indexing problem} is to preprocess an $N$-length vector $u$ of real numbers and a set of indices $S$ to support the following queries: given an $M$-length vector $v$ ($M<N$) of real numbers, for each $i \in S$ compute the value of the $i$-th element of the convolution of $u$ and $v$.
\end{definition}

Once again this variant already relevant
when the input is Boolean and sparse, i.e. $u$ and $v$ have $n=O(N^{1-\Omega(1)})$ ones and are represented implicitly by specifying their indices.%. , as otherwise there is nothing to gain from an implicit representation. \tknote{Please verify this last sentence. Also, I feel like we need to explain this better.}

\medskip
\noindent
We prove the following \CLB{s} for these problems with the help of our framework.

\begin{introthm}
Assume the \ThreeSUM{} conjecture is true.
Then there is no algorithm for the partial convolution problem with $O(N^{1-\Omega(1)})$ time, even if $|S|$ and the number of ones in both input vectors are less than $N^{1-\Omega(1)}$.
\end{introthm}

\begin{introthm}
Assume the \ThreeSUM{} conjecture is true.
Then there is no algorithm for the partial convolution indexing problem with $O(N^{2-\Omega(1)})$ preprocessing time and $O(N^{1-\Omega(1)})$ query time, even if both $|S|$ and the number of ones of the input vectors are $O(N^{1-\Omega(1)})$.
\end{introthm}

As mentioned above, the convolution of vectors of length $N$ can be computed in $\tilde{O}(N)$ time with the FFT algorithm. However, in the partial convolution problem and partial convolution indexing problem,
despite the input vectors being sparse and represented sparsely (specifying only the $O(N^{1-\Omega(1)})$ indices of the ones in each vector), and despite the portion of the output we need to compute being sparse ($|S| = O(N^{1-\Omega(1)})$), no linear time algorithm (in $n=O(N^{1-\Omega(1)})$) exists, unless the \ThreeSUM{} conjecture is false.%These \CLB{s} hold also for the special case of Boolean vectors, and, therefore, also for the special case in which we only want to know if the output values at indices in $S$ are zero or more.

Notice that the partial convolution problem and its indexing variant are {\em decision} problems, since they require a {\em decision} for each location $i\in S$, whether $w[i]>0$ or not. This is in contrast to the convolution witnesses problem, which is a reporting problem, as it requires the reporting of \emph{all} of the witnesses for $w[i]$.

To prove \CLB{s} for the convolution problems we follow our framework. That is, we first use a hash function to embed a \ThreeSUM{} instance to a \emph{smaller} universe. This mapping introduces false-positives, which we enumerate by utilizing the reporting problem of convolution witnesses. To solve the reporting version we reduce it to several instances of a decision problem, partial convolution or its indexing variant, by constructing a suitable data structure. Tying it all together leads to \CLB{s} for both the reporting and decision problems.

\vspace{5pt}
\noindent
 {\bf Matrix Problems\ }

  \noindent
We also present some similar \CLB{s} for matrices.

\begin{definition} The {\bf partial matrix multiplication  problem} on two $N \times N$ matrices $A$ and $B$ of real numbers
and a set of entries $S \subseteq N \times N$ is to compute, for each $(i,j)\in S$, the value $(A\times B)[i,j]$.
\end{definition}

The indexing variant of this problem is defined as follows.

\begin{definition}
The {\bf partial matrix multiplication indexing problem} is to preprocess an $N \times N$ matrix $A$ of real numbers and a collection $\mathcal{S}=\{S_1,S_2,...,S_k\}$ of sets of entries, where $S_i \subseteq N \times N$, so that given a sequence $B_1,\ldots,B_k$ of $N \times N$ matrices of real numbers, we enumerate the entries of $A\times B_i$ that correspond to $S_i$.
\end{definition}

For $\mathcal{S}=\{S_1,S_2,...,S_k\}$ let $SIZE(S)=\sum_{i=1}^{k} |S_i|$.
We prove the following \CLB{s}, which hold also for the special case of Boolean multiplication assuming that the input is given implicitly by specifying only the indices of the ones.

\begin{introthm}\label{thm:partial-matrix-mult}
Assume the \ThreeSUM{} conjecture is true.
Then there is no algorithm for the partial matrix multiplication problem running in $O(N^{2-\Omega(1)})$ expected time, even if $|S|$ and the number of ones in the input matrices is $O(N^{2-\Omega(1)})$.
\end{introthm}

\begin{introthm}\label{thm:partial-matrix-mult-indexing}
Assume the \ThreeSUM{} conjecture is true.
Then there is no algorithm for the partial matrix multiplication indexing problem with $O(SIZE(S))$ preprocessing time and $O(N^{2-\Omega(1)})$ query time.
\end{introthm}

Matrix multiplication, and in particular Boolean matrix multiplication, can be solved in $\tilde{O}(n^\omega)$ time, where $\omega \approx 2.373$~\cite{Gall14, Williams12}. Many researchers believe that the true value of $\omega $ is $2$. This belief implies that the running time for computing the product of two Boolean matrices is proportional to the size of the input matrices and the resulting output. However, our results demonstrate that such a result is unlikely to exist for sparse versions of the problem, where the number of ones in the matrices is $O(N^{2-\Omega(1)})$ and we are interested in only a partial output matrix (only $O(N^{2-\Omega(1)})$ entries of the matrix product).

To prove Theorem~\ref{thm:partial-matrix-mult} and~\ref{thm:partial-matrix-mult-indexing} we follow our framework. The process is very similar to the path for proving \CLB{s} for convolution problems. In fact, instead of considering a reporting version of the partial matrix multiplication problem for proving these \CLB{s}, we once again utilize the reporting problem of convolution witnesses. However, this time we transform the convolution witnesses to the matrix multiplication problems using a more elaborate data structure. The main difficulty in this transformation is to guarantee the sparsity of both the input and the required output. This transformation illustrates how a reporting version of a problem can be used to prove \CLB{s} for decision versions of other problems, by changing the way we look for honest witnesses.

\vspace{5pt}

\noindent
 {\bf String Problems\ }

  \noindent
Another application of our framework, which is seemingly unrelated to the previous two, is the problem of histogram indexing. A {\em histogram}, also called a {\em Parikh vector}, of a string $T$ over alphabet $\Sigma$ is a $|\Sigma|$-length vector containing the character count of $T$. For example, for $T=abbbacab$ the histogram is $\psi(T)=(3,4,1)$.

\begin{definition}
In the {\bf histogram indexing problem} we preprocess a string $T$ to support the following queries: given a query Parikh vector $\psi$, return whether there is a substring $T'$ of $T$ such that $\psi(T')=\psi$.
\end{definition}

\begin{definition}
In the {\bf histogram indexing reporting problem} we preprocess a string $T$ to support the following queries: given a query Parikh vector $\psi$, report indices of $T$ at which a substring $T'$ of $T$ begins such that $\psi(T')=\psi$.
\end{definition}

The problem of histogram indexing (not the reporting version) is sometimes called {\em jumbled indexing}. It has received much attention in recent years. For example, for binary alphabets - that is histograms of length 2 - there is a straightforward algorithm with $O(n^2)$ preprocessing time and constant query time, see \cite{CFL09}. Burcsi et al. \cite{BCFL12} and  Moosa and Rahman \cite{MR10} improved the preprocessing time to $O(n^2/\log{n})$. Using the four-Russian trick a further improvement was achieved by Moosa and Rahman \cite{MR12}. Then, using a connection to the recent improvement of all-pairs-shortest path by Williams~\cite{Williams14}, as observed by Bremner et al.~\cite{BCDEHILPT14} and by Hermelin et al. \cite{HLRW14}, the preprocessing time was further reduced to $O(\frac{n^2}{2^{\Omega(\log{n})^{0.5}}})$ . Finally, Chan and Lewenstein~\cite{CL15} presented an $O(n^{1.859})$ preprocessing time algorithm for the problem with constant query time. For non-binary alphabets some progress was achieved in the work by Kociumaka et al. \cite{KRR13} and even further achievement was shown in~\cite{CL15}. On the negative side, some \CLB{s} were  recently shown by Amir et al. \cite{ACLL14}.

We follow our framework and first obtain \CLB{s} for the reporting version of histogram indexing. This is the first time \CLB{s} are shown for the reporting version. Moreover, these \CLB{s} apply to binary alphabets, as opposed to the decision version in which there currently is no \CLB{} known for binary alphabets. The \CLB{s} for the reporting version admit a full tradeoff between preprocessing and query time. For the decision variant, we improve upon the \CLB{} by Amir et al.~\cite{ACLL14} by presenting full-tradeoffs between preprocessing and query time based on the standard \ThreeSUM{} conjecture. Specifically, our new \CLB{} implies that no solution to the histogram indexing problem exists with subquadratic preprocessing time and $\tilde{O}(1)$ query time for every alphabet size bigger than 2, unless the 3SUM conjecture is false. This demonstrates a sharp separation between binary and trinary alphabets, since Chan and Lewenstein~\cite{CL15} introduced an algorithm for histogram indexing on binary alphabets with $\tilde{O}(n^{1.859})$ preprocessing time and constant query time.
A complete comparison of our results and the results by Amir et al. \cite{ACLL14} appears in Table~\ref{table:histogram_indexing_comparison} in Appendix~\ref{sec:summary_of_applications}.

The \CLB{s} are summarized by the following theorems.

\begin{introthm}
Assume the \ThreeSUM{} conjecture is true.
Then the histogram reporting problem for an $N$-length string and constant alphabet size $\ell \geq 2$ cannot be solved using $O(N^{2-\frac{2\gamma}{\ell+\gamma}-\Omega(1)})$ preprocessing time, $O(N^{1-\frac{\gamma}{\ell+\gamma}-\Omega(1)})$ query time and $O(N^{\frac{\gamma \ell}{\ell+\gamma} - \frac{2\gamma}{\ell+\gamma}-\Omega(1)})$ reporting time per item, for any $0 < \gamma < \ell$.
\end{introthm}

\begin{introthm}
Assume the \ThreeSUM{} conjecture holds.
Then the histogram indexing problem for a string of length $N$ and constant alphabet size $\ell \geq 3$ cannot be solved with $O(N^{2-\frac{2(1-\alpha)}{\ell-1-\alpha}-\Omega(1)})$ preprocessing time and $O(N^{1-\frac{1+\alpha(\ell-3)}{\ell-1-\alpha}-\Omega(1)})$ query time.
\end{introthm}

The main structure of these proofs follows our framework. We first embed a \ThreeSUM{} instance and encode it in a string with limited length. We then report the false-positives using the reporting variant of the histogram indexing problem, which implies \CLB{s} for this variant. Finally, we reduce the reporting version to the decision version thereby obtaining \CLB{s} for the decision version. The reduction utilizes a sophisticated data structure for reporting witnesses using many instances of the decision version.

\section{Preliminaries}\label{sec:preliminaries}
In the basic \ThreeSUM{} problem we are given a set $A$ of $n$ integers and we need to answer whether there are $a,b,c \in A$ such that $a+b+c=0$. In a common variant of the classic problem, which we also denote by \ThreeSUM{}, three arrays $A, B$ and $C$ are given and we need to answer whether there are $a \in A,b \in B ,c \in C$ such that $a+b+c=0$. Both versions have the same computational cost (see~\cite{GO95}).
There are some other variants of the \ThreeSUM{} problem shown to be as hard as \ThreeSUM{} up to poly-logarithmic factors.
One such variant is \ConvolutionThreeSUM{}, shown to be hard by \Patrascu{} \cite{Patrascu10}, see also~\cite{KPP16}.
In \ConvolutionThreeSUM{} $A$ is an ordered set and we need to answer whether there exist indices $0 \leq i,j \leq n-1$ such that $A[i]+A[j]=A[i+j]$. We also define \DiffConvolutionThreeSUM{}, in which we are given an ordered set $A$ and we need to verify whether there exists $0 \leq i,k \leq n-1$ such that $A[k]-A[i]=A[k-i]$. It is easy to see that this is equivalent to \ConvolutionThreeSUM{}.

\remove{Let $\mathcal{H}$ be a family of hash functions from $[u]$ to $[m]$.
%$\mathcal{H}$ is called {\em $c$-universal} if for any distinct $x,x' \in [u]$, we have
%$\Pr_{h\in \mathcal{H}}(h(x) = h(x')) \le \frac{c}{m}$.
$\mathcal{H}$ is called {\em linear} if for any $h\in\mathcal{H}$ and any $x,x' \in [u]$, we have $h(x) + h(x') \equiv h(x+x') \; (\modulo m)$.
$\mathcal{H}$ is called {\em almost-linear} if for any $h\in\mathcal{H}$ and any $x,x' \in [u]$, we have
either $h(x) + h(x') \equiv h(x+x') \; (\modulo m)$, or $h(x) + h(x') \equiv h(x+x') +c_h \; (\modulo m)$, where $c_h$ is an integer that depends only on the choice of $h$.
For a function $h:[u] \rightarrow [m]$ and a set $S\subset [u]$ where $|S|=n$, we say that $i\in [m]$ is an {\em overflowed value} of $h$ if $|\{x\in S : h(x) = i\}| > 3n/m$.
$\mathcal{H}$ is called {\em almost-balanced} if for a random $h\in \mathcal{H}$ and any set $S\subset [u]$ where $|S|=n$, the expected number of elements from $S$ that are mapped to overflowed values is $O(m)$.
Baran et al. \cite{BDP05} showed that the family of hash functions defined by Dietzfelbinger \cite{Dietzfelbinger96} is almost-linear and almost-balanced.
}

Let $\mathcal{H}$ be a family of hash functions from $[u] \rightarrow [m]$.
%$\mathcal{H}$ is called {\em $c$-universal} if for any distinct $x,x' \in [u]$, we have
%$\Pr_{h\in \mathcal{H}}(h(x) = h(x')) \le \frac{c}{m}$.
$\mathcal{H}$ is called {\em linear} if for any $h\in\mathcal{H}$ and any $x,x' \in [u]$, we have $h(x) + h(x') \equiv h(x+x') \; (\modulo m)$.
$\mathcal{H}$ is called {\em almost-linear} if for any $h\in\mathcal{H}$ and any $x,x' \in [u]$, we have
either $h(x) + h(x') \equiv h(x+x') +c_h \; (\modulo m)$, or $h(x) + h(x') \equiv h(x+x') + c_h +1 \; (\modulo m)$, where $c_h$ is an integer that depends only on the choice of $h$.
For a function $h:[u] \rightarrow [m]$ and a set $S\subset [u]$ where $|S|=n$, we say that $i\in [m]$ is an overflowed value of $h$ if $|\{x\in S : h(x) = i\}| > 3n/m$.
$\mathcal{H}$ is called {\em almost-balanced} if for a random $h\in \mathcal{H}$ and any set $S\subset [u]$ where $|S|=n$, the expected number of elements from $S$ that are mapped to overflowed values is $O(m)$.
See~\cite{KPP16} for constructions of families that are almost-linear and almost-balanced (see also \cite{BDP05,Dietzfelbinger96}).

For simplicity of presentation, and following the footsteps of previous papers that have used such families of functions~\cite{BDP05,Patrascu10}, we assume for the rest of the paper that almost linearity implies that for any $h\in\mathcal{H}$ and any $x,x' \in [u]$ we have
$h(x) + h(x') \equiv h(x+x') \; (\modulo m)$. There are actually two assumptions taking place here. The first is that there is only one option of so-called linearity. Overcoming this assumption imposes only a constant factor overhead. The second assumption is that $c_h=0$. However, the constant $c_h$ only affects offsets in our algorithm in a straightforward and not meaningful way, so we drop it in order to avoid clutter in our presentation.

%% file: conv_witnesses.tex
\section{Convolution Witnesses}
We first prove a \CLB{} for the convolution witnesses problem. We begin with a lemma which has elements from the proof of \Patrascu's reduction~\cite{Patrascu10} and from~\cite{BDP05}. However, the lemma diverges from~\cite{Patrascu10} by treating the hashed subsets differently. Specifically, many special \ThreeSUM{} subproblems are created and then reduced to convolution witnesses.

We say that a binary vector of length $n$ is {\em $r$-sparse} if it contains at most $r$ $1$'s.
An instance of convolution witnesses problem $(u,v,w)$ is $(n,R)$\emph{-sparse} if $u$ and $v$ are both of length $n$ and $n/R$-sparse.

\begin{mylemma}\label{lem-convolution}
Let sequence $A=\langle x_1,\cdots, x_n \rangle$ be an instance of \ConvolutionThreeSUM{}. Let $R = O(n^{\delta})$, where $0<\delta<0.5$ is a constant. There exists a truly subquadratic reduction from the instance $A$ to $O(R^2)$ $(n,R)$-sparse instances of convolution witnesses problem %$(u_i,v_i,w_i)$ ($1 \leq i \leq R^2$)
for which we need to report $O(n^2/R)$ witnesses (over all instances).
\end{mylemma}

\begin{proof}
We use an almost-linear, almost-balanced, hash function $h:U\rightarrow [R]$ and create $R$ buckets $B_0,\cdots, B_{R-1}$ where each $B_a$ contains the indices of all elements $x_i \in A$ for which $h(x_i)=a$.
Since $h$ is almost-balanced the expected overall number of elements in buckets with more than $3n/R$ elements is $O(R)$. For each index $i$ in an overflowed bucket, we verify whether $x_i + x_j = x_{i+j}$ for every other $j$ in $O(n)$ time. Hence, we verify whether any index in an overflowed bucket is part of a \ConvolutionThreeSUM{} solution in $O(nR)$ expected time. Since $R = O(n^{1-\Omega(1)})$ the expected time is truly subquadratic time. 

We now assume that every bucket contains at most $3n/R$ elements. From the properties of almost-linear hashing, if $x_i+x_j=x_{i+j}$ then $h(x_i)+h(x_j) \modulo R =  h(x_{i+j }) \modulo R$. Hence, if $x_i+x_j=x_{i+j}$  then $i\in B_a, j\in B_b$ implies that $i+j\in B_{a+b \modulo R}$.

Every three buckets form an instance of \ThreeSUM{} and are uniquely defined by $a$ and $b$. Hence, there are $R(R-1)/2=O(R^2)$  \ThreeSUM{} subproblems each on $O(n/R)$ elements from the small universe $[n]$. 
However, $h$ may generate false positives. So, we must be able to verify that any \ThreeSUM{} solution (a witness) for any instance is indeed a solution (an honest witness) for the problem on $A$.
The number of false positives is expected to be $O(n^2/R)$ over all $O(R^2)$ instances, see~\cite{BDP05}. So, we need an efficient tool to report each such witness in order to be able to solve \ConvolutionThreeSUM{}.

To obtain such a tool, we reduce the problem to the convolution setting in the following way. We generate a characteristic vector $v_a$ of length $n$ for every set $B_a$ ($v_a[i]=1$ if $i \in B_a$ and $v_a[i]=0$ otherwise, for $0 \leq i < n$). This vector will be $3n/R$-sparse, since $|B_a|\leq 3n/R$. Note that: $i\in B_a, j\in B_b\rm{\ and\ } i+j  \in B_{a+b\modulo R} \iff v_a[i]=1, v_b[j]=1 \rm{\ and\ } v_{a+b\modulo R}[i+j]=1$.
Now, for each pair of vectors, $v_a$ and $v_b$, we generate their convolution. Let $v=v_a * v_b$ be the convolution of $v_a$ and $v_b$, and let $\ell=v[i+j]$. If $v_{a+b\modulo R}[i+j]=1$, then we need to extract the $\ell$ witnesses of $v[i+j]$. For each witness $(i,j)$ we check whether $x_i+x_j=x_{i+j}$. We note that if, while verifying, we discover that the overall number of the false-positives exceeds expectation ($cn^2/R$, for some constant $c$) by more than twice we rehash.

Thus, we see that \ConvolutionThreeSUM{} can be solved by generating $O(R^2)$ $(n,R)$-sparse instances of convolution witnesses problem. These instances are computed in $O(nR^2)$ time, which is truly subquadratic as $R = O(n^{\delta})$ for $\delta<1/2$.
\end{proof}

\noindent
It now follows that:

\begin{theorem}~\label{thm-convolutiuon}
Assume the \ThreeSUM{} conjecture is true.
Then for any constant $0<\alpha <1$, there is no algorithm solving the convolution witnesses problem with $O(n^{2-\alpha})$ expected preprocessing time and $O(n^{\alpha/2 - \Omega(1)})$ expected amortized query time {\bf per witness}.
\end{theorem}

\begin{proof}
We make use of Lemma~\ref{lem-convolution} and its parameter $R$. In particular, the total cost of solving \ConvolutionThreeSUM{} is at most $O(R^2 \cdot P(n,R) + n^2/R \cdot Q(n,R))$ expected time, where $P(n,R)$ is the time needed to preprocess an $(n,R)$-sparse instance of a convolution witness and $Q(n,R)$ is the time per witness query for an $(n,R)$-sparse instance of a convolution witness.

If we choose $R=n^{\alpha /2-\Omega(1)}$ we have that for $P(n)=O(n^{2-\alpha})$ and $Q(n)=O(n^{\alpha/2 - \Omega(1)})$ we
solve \ConvolutionThreeSUM{} in $O(n^{2-\Omega(1)})$ time which is truly subquadratic.
\end{proof}

%% file: partial_conv.tex
\section{From Reporting to Decision I: Hardness of Partial Convolutions}~\label{reportingViaDecision}
We further consider the problem of reporting witnesses for convolutions. However, now we use the third step of our framework. We will construct a search data structure over decision problems which will allow us to efficiently search for witnesses. This will be our method for proving \CLB{s} for the decision problems of partial convolutions~\cite{Muthukrishnan95}. Specifically, we intend to generate a data structure that uses convolutions on small sub-vectors of the input vectors in order to solve the problem.  However, the data structure cannot be fully constructed as it will be too large. Hence, the construction is partial and we defer some of the work to the query phase.

We start with Lemma~\ref{lem-convolution}, and focus on an $(n,R)$-sparse instance of the convolution witnesses problem$(u,v,w)$. We generate a specialized search tree for efficiently finding witnesses, which is created in an innovative way exploiting the sparsity of the input.

\subsection{Search Tree Construction}~\label{search-tree-construction}
Assume, without loss of generality, that $n$ is a power of $2$. We construct a binary tree in the following way. First, we generate the root of the tree with the convolution of $v$ and $u$. Then we split $u$ into $2$ sub-vectors, say $u_1$ and $u_2$, each containing exactly $n/(2R)$ $1$s. For each sub-vector we generate nodes that are children of the root, where the first node contains the convolution of $v$ and $u_1$ and the second node contains the convolution of $v$ and $u_2$. We continue this construction recursively so that at the $i$th recursive level we partition $u$ into $2^i$ sub-vectors each containing $n/(2^iR)$ $1$s. A vertex at level $i$ represents the convolution of $v$ and a sub-vector $u_A$ containing $n/2^iR$ $1$s. The vertex has two children, one represents the convolution of $v$ and the sub-vector of $u_A$ with the first $n/2^{i+1}R$ $1$s of $u_A$ (denoted by $u_{A,1}$). The other represents the convolution of $v$ and the rest of $u_A$ with the other $n/2^{i+1}R$ $1$s (denoted by $u_{A,2}$).
We stop the construction at the leaf level in which $u$ is split to sub-vectors that each one of them contains $X/R$ $1$s from $u$, for some $X<n$ to be determined later. Calculating the convolution in each vertex is done bottom-up. First, we calculate the convolution for each vertex in the leaf level. Then, we use these results to calculate the convolution of the next level upwards. Specifically, if we have vertex that represent the convolution $v$ and some sub-vector $u_A$ and it has two children   one which represents the convolution of $v$ and $u_{A,1}$ and the other which represents the convolution of $v$ and $u_{A,2}$, then $(v * u_A)[k] = (v * u_{A,1})[k] + (v * u_{A,2})[k-l_1]$ for every $k \in [0,n+l_1+l_2-1]$, where $l_1$ and $l_2$ are the lengths of $u_{A,1}$ and $u_{A,2}$ respectively, and we consider the value of out of range entries as zero. This way we continue to calculate all the convolutions in the tree until reaching its root.

\smallskip
\noindent
{\bf Construction Time.\ } It is straightforward to verify that the total cost of the construction procedure is dominated by the time of constructing the lowest level of the binary tree. In this level, we have $n/X$ sub-vectors of $u$ as each of them has $X/R$ 1's and the total number of $1$s in $u$ is $n/R$. We calculate the convolution of $v$ with each of these sub-vectors, which can be done in $\tilde{O}(n)$ time. Thus, the total time needed to build the tree is $\tilde{O}(n^2/X)$.
Therefore, the total time for calculating the binary trees for all $O(R^2)$ $(n,R)$-sparse instances of the convolution witnesses problem is $\tilde{O}(R^2n^2/X)$.

\smallskip
\noindent
{\bf Witness Search.\ }
To search for a witness we begin from the root of the binary tree and traverse down to a leaf containing a non-zero value in the result of the convolution at the query index (adjusting the index as needed while moving down the structure). The search for a leaf costs logarithmic time per query (as the tree has logarithmic height and in each level we just need to find a child with a non-zero value in the convolution it represents in the specific index of interest). Within the leaf, representing the convolution of $v$ and some sub-vector $u_A$ of $u$ we can simply find a witness in $\tilde{O}(X/R)$ time as $u_A$ contains just $X/R$ $1$s. Thus, as we have $O(n^2/R)$ false-positives over all $O(R^2)$ instances, the total time for finding all them is $\tilde{O}(n^2X/R^2)$.

\smallskip
Consequently, using the binary tree for solving \ConvolutionThreeSUM{} will cost $\tilde{O}(R^2n^2/X + n^2X/R^2)$ time, which for $X=R^2$ is $\tilde{\Theta}(n^2)$ time.
Since the tradeoff between the preprocessing time and query time meets at $n^2$, any improvement to the running time of either of them will imply a subquadratic solution for the \ConvolutionThreeSUM{} problem.

\subsection{Conditional Lower Bounds for Partial Convolution}

As a consequence of our discussion above we obtain the following results regarding partial convolution and its indexing variant:

\begin{theorem}\label{thm:partial-convolution}
Assume the \ThreeSUM{} conjecture is true.
Then there is no algorithm for the partial convolution problem with $O(N^{1-\Omega(1)})$ time, even if $|S|$ and the number of ones in both input vectors are less than $N^{1-\Omega(1)}$.
\end{theorem}
\begin{proof}
We make use of Lemma~\ref{lem-convolution}. In order to construct the binary tree as described in Section~\ref{search-tree-construction}, we need to be able compute the convolution of $v$ with some sub-vector of $u$ for each leaf in the tree (all other convolution can be calculated efficiently from the convolutions in the leaves as described in the previous section). Recall that both input vectors have length $N=n$, $n/R$ $1$s (which is $O(N^{1-\Omega(1)})$ for $R=n^a$, where $a$ is a positive constant), and we are interested in finding their convolution result only at the $O(n/R)$ indices (that is, $|S| = O(N^{1-\Omega(1)})$). If we preprocess the input for partial convolution in truly sublinear time (for example, proportional to $n/R$) then the total time for constructing all the search trees will be $O(R^2n^{2-\Omega(1)}/X)$ while the total query time will remain $O(n^2X/R^2)$. Choosing $X=n^c$ for small enough constant $c$ and setting $R=X$, we obtain a subquadratic solution to \ConvolutionThreeSUM{}.
\end{proof}

\begin{theorem}\label{thm:partial-convolution-indexing}
Assume the \ThreeSUM{} conjecture is true.
Then there is no algorithm for the partial convolution indexing problem with $O(N^{2-\Omega(1)})$ preprocessing time and $O(N^{1-\Omega(1)})$ query time, even if both $|S|$ and the number of ones of the input vectors are $O(N^{1-\Omega(1)})$.
\end{theorem}

\begin{proof}
Use Lemma~\ref{lem-convolution} and the previous discussion. If the preprocessing time for the partial convolution
indexing problem is truly subquadratic and queries are answered in truly sublinear time then the total time for constructing all the structures for all $O(R^2)$ instances is $O(R^2[n^{2-\Omega(1)} + n^{1-\Omega(1)} \cdot n/X])$ while the total time for all of the queries remains $O(n^2X/R^2)$ (note that $N=n$).
Choosing $X=n^c$ for small enough constant $c$ and setting $R=X$, we obtain a subquadratic algorithm for \ConvolutionThreeSUM{}.
\end{proof}

\subsection{Search Tree Construction}~\label{search-tree-construction}

Notice that the convolution of $u$ and $v$ immediately provides the number of witnesses for each index in linear time (since the vectors are binary).

Assume, without loss of generality, that $n$ is a power of $2$. Let $X<n$ be a power of 2 to be determined later. For every $0\leq i \leq \log n-\log X$ we split the two input vectors to $2^i$ sub-vectors of size $n/2^i$, and for each pair of sub-vectors, one from $u$ and one from $v$, we compute their convolution.
We construct a quad-tree such that each vertex contains the convolution between two sub-vectors. Each vertex $x$ (besides the leaves) has exactly 4 children, defined in the following way. If $x$ corresponds to the convolution between two sub-vectors $u_A$ and $v_B$ of length $n/2^i$, then its children correspond to the convolutions between the sub-vectors of length $n/2^{i+1}$ that are the halves of the sub-vectors $u_A$ and $v_B$.

Formally, the convolution vector is defined to be $w[k]= \sum_{j}{u[j]v[k-j]}$. This is defined for all $k\in[0,2n-2]$. Consider a vertex that represents sub-vectors $u_A$ and $v_B$ of length $n/2^i$, where $A = [a,a+n/2^i-1]$ and $B = [b,b+n/2^i-1]$, i.e. $u_A=u[a]u[a+1]\cdots u[a+n/2^i-1]$. Then, $u[r]=u_A[r-a]$ and $v[r]=v_B[r-b]$ (for $a \leq r < a+n/2^i$). Hence, $w[k]$ restricted to $A$ and $B$ is\\ $\sum_{j}{u_A[j-a]v_B[k-j-b]} = \sum_{j}{u_A[j-a]v_B[k-b-a-(j-a)]}=w_{A,B}[k-a-b]$. For most $A$ and $B$ the value $w_{A,B}[k-a-b]$ is out of scope, that is $k-a-b \notin [0,2(n/2^i)-2]$. It is easy to verify that summing over the sub-vectors of size $n/2^i$ that are in scope for $k$, $w[k]=\sum_{A,B}w_{A,B}[k-a-b]$. Hence, if there is a witness for $w[k]$ then we find it by finding a witness for $w_{A,B}$ where $w_{A,B}[k-a-b]$ is non-zero.

\bigskip
\noindent
{\bf Construction Time.\ } Every leaf in the quad-tree represents a pair $(u_A,v_B)$ where $|A|=|B|=X$. Since the constructed tree is a complete quad-tree, it is straightforward to verify that the total cost of the construction procedure is dominated by the time of constructing the lowest level of the quad-tree. In this level, the length of the convoluted sub-vectors is $X$. Thus, the total time needed to build the tree is $\tilde{O}(n^2/X)$ since we compute the convolution between every pair of sub-vectors (one from $u$ and one from $v$) of length $X$ (we have $(n/X)^2)$ pairs), and each such convolution costs $X\log X$ time.
Thus, the total time for calculating the quad trees for all $O(R^2)$ instances of the $(n,R)$-sparse convolution witnesses problem is $O(R^2n^2/X)$.

\bigskip
\noindent
{\bf Witness Search.\ }
To search for a witness we begin from the root of the quad-tree and traverse down to a leaf containing a non-zero value in the result of the convolution at the query index (adjusting the index as needed while moving down the structure). The search for a leaf costs logarithmic time per query (as the tree has logarithmic height and in each level we just need to find a child with non-zero value in the convolution it represents in the specific index of interest). Within the leaf, representing a pair of sub-vectors $A$ and $B$, we compute all witnesses of $w_{A,B}[k-a-b]$ in $\tilde{O}(X)$ time. Thus, the total time for finding all the false-positives is $\tilde{O}(n^2X/R)$.

Consequently, using the quad-tree for solving \ConvolutionThreeSUM{} will cost $\tilde{O}(R^2n^2/X + n^2X/R)$ time. If we choose $X=R^{1.5}$ we balance the two and obtain $n^2R^{0.5}$ which is $\omega(n^2)$ time, which is too high for obtaining a lower bound from the \ThreeSUM{} conjecture. However, we now refine this idea.

\bigskip
\noindent
{\bf Improved Construction.\ }
In order to obtain a better total running time we use a different construction as follows. For every $0\leq i \leq \log n-\log X$ we only split $u$ into $2^i$ sub-vectors of size $n/2^i$, and compute the convolution between $v$ and each sub-vector of $u$. Hence, instead of a quad-tree we obtain a binary tree. In this tree, each vertex corresponds to the convolution of $v$ and a sub-vector $u_A$ of $u$ of length $n/2^i$ and has as its children the 2 vertices corresponding to convolutions between $v$ and the sub-vectors of length $n/2^{i+1}$ that are the halves of $u_A$.
Specifically, as before we define $u_A=u[a]u[a+1]\cdots u[a+n/2^i-1]$. Moreover, $w[k]= \sum_{i}{u[i]v[k-i]}$ restricted to $A$ (and all of $v$) is $\sum_{i}{u_A[i-a]v[k-i]} = \sum_{i}{u_A[i-a]v[k-a-(i-a)]}=w_{A}[k-a]$.

The total cost of building this tree is $O(n^2/X)$ time.
The search for a leaf costs logarithmic time per query. Searching within the leaf will still cost $O(X)$ time in order to find all witnesses for a specific index of the convolution it represents by using the naive algorithm. This still requires $\omega(n^2)$ running time as before, even for the best choice of $R$. However, this is a step toward the next refinement that gives us the desired improvement.

\bigskip
\noindent
{\bf Even Further Improvement.\ }
We further reduce the search time by splitting the vector into (unequal sized) sub-vectors of $u$ during the construction phase according to the location of the 1s in $u$ as opposed to sub-vector lengths. Recall that there are $O(n/R)$ $1$'s in $u$. For the sake of clarity, assume that there are exactly $n/R$ $1$'s in $u$. First we generate the root of the binary tree with the convolution of $u$ and $v$. Then we split $u$ into $2$ sub-vectors, say $u_1$ and $u_2$, each containing exactly $n/2R$ $1$'s. For each sub-vector we generate nodes that are children of the root, where the first node contains the convolution of $u_1$ and $v$ and the second node contains the convolution of $u_2$ and $v$. We continue this construction recursively so that at the $i$'th recursive level we partition $u$ into $2^i$ sub-vectors each containing $n/2^iR$ $1$'s. A vertex at level $i$ represents the convolution of $v$ and a sub-vector $u_A$ containing $n/2^iR$ $1$'s. The vertex has two children, one which represents the convolution of $v$ and the sub-vector of $u_A$ with the first $n/2^{i+1}R$ $1$'s of $u_A$ and the other which represents the convolution of $v$ and the rest of $u_A$ with the other $n/2^{i+1}R$ $1$'s.

The leaf level contains sub-vectors $A$ with $X/R$ $1$'s from $u$. Hence, the computation of all witnesses of $w_{A}[k-a]=\sum_{i}{u_A[i-a]v[k-a-(i-a)]}$ only requires access to at most $X/R$ $i$'s for which $u_A[i-a]=1$, thereby taking $O(n/X)$ time. Thus, the total time cost of this refined solution is $O(R^2n^2/X + n^2X/R^2)$ which for $X=R^2$ is $\Theta(n^2)$ time.
Since the tradeoff between the preprocessing time and query time meets at $n^2$, any improvement to the running time of either of them will imply a subquadratic solution for the \ConvolutionThreeSUM{} problem.

%% file: partial_matrix_mult.tex
\section{From Reporting to Decision II: Hardness of Partial Matrix Multiplication}\label{sec:partial_matrix_mult}

Using ideas of the same flavor as those from Section~\ref{reportingViaDecision} we prove \CLB{s} on the partial matrix multiplication problem, which is defined in the introduction.

We start with Lemma~\ref{lem-convolution} which shows how to reduce an instance of \ConvolutionThreeSUM{} to $O(R^2)$ instances of convolution witnesses problem. We will prove the \CLB{} by showing that partial matrix multiplication can be used to solve a major component of the computation necessary. We focus, again, on two binary vectors $u$ and $v$ of size $n$ (out of the $O(R^2)$ pairs of vectors), each containing $O(n/R)$ $1$s, and their convolution result $w$. Similar to Section~\ref{reportingViaDecision}, we construct a search tree to seek for witnesses. However, this time we partition each of the vectors $u$ and $v$ into $\Theta(n/X)$ sub-vectors, and the partitioning method here slightly differs from the one in Section~\ref{reportingViaDecision}. The main obstacle is to guarantee the spareness of the input and the required output of the matrices we will use in our construction.

The tree we construct for efficiently searching for witnesses is computed as follows.

\bigskip
\noindent
{\bf Quad Tree Construction with Special Leaves\ }
First partition each vector to $n/X$ sub-vectors with each sub-vector containing $O(X/R)$ $1$s. Then, partition each sub-vector whose length is more than $X$ into smaller sub-vectors with length exactly $X$ except the last one that might be shorter. Pad each sub-vector with $0$s, if necessary, so that all lengths are exactly $X$. Denote the sub-vectors $u_1,u_2,...,u_q$. It is straightforward to observe that:

\begin{observation}
$\Sigma_{i=1}^q |u_i| = O(n)$. That is, $q= cn/X$.
\end{observation}

Hence, the partition of $u$ is to (padded) sub-vectors $u_1,u_2,...,u_{cn/X}$, each of length $X$ and with $O(X/R)$ $1$s each. The same process is done for vector $v$ and we get sub-vectors $v_1,v_2,...,v_{c'n/X}$ that satisfy the same properties. We assume without loss of generality that $c=c'$ and that $cn/X$ is a power of two.

We construct a quad tree instead of the binary tree from Section~\ref{search-tree-construction}.
For integers $1\leq i<j$ denote by $u_{i,j}$ the sub-vector of $u$ that is the concatenation of the sub-vectors $u_i,u_{i+1}, \ldots, u_j$. Similarly, for integers $1\leq i<j$ denote by $v_{i,j}$ the sub-vector of $v$ that is the concatenation of the sub-vectors $v_i,v_{i+1}, \ldots, v_j$. The root of the quad tree will contain the result of the convolution (which we compute later) of $u$ and $v$, which are $u_{1,\frac{cn}{X}}$ and $v_{1,\frac{cn}{X}}$. The root has four children which correspond to all the convolutions between the 4 pairs of sub-vectors in $\{u_{1,\frac{cn}{2X}}, u_{\frac{cn}{2X}+1,\frac{cn}{X}}\} \times \{v_{1,\frac{cn}{2X}}, v_{\frac{cn}{2X}+1,\frac{cn}{X}}\}$.
Recursively, a vertex representing the convolution of a pair $u_{i,j}$ and $v_{i',j'}$ has 4 children representing the convolutions of
the 4 sub-vector pairs in $\{u_{i,i+(j-i)/2}, u_{(i+(j-i)/2)+1,j}\} \times \{v_{i',i'+(j'-i')/2}, v_{(i'+(j'-i)/2)+1,j'}\}$.

Our goal is to compute for each vertex in this tree the convolution of the two sub-vectors that it covers. In Section~\ref{reportingViaDecision} we directly computed the convolution for every node. Here, we will use matrix multiplication in order to achieve our reduction. We do the computation on the leaves of the tree (we shortly explain how) and then use a bottom-up traversal of the tree where we compute the convolution for an inner vertex from the convolutions of its children (without computing it directly).

\bigskip
\noindent
{\bf Seeking a Witness\ }
This is literally done in the same way as described in Section~\ref{search-tree-construction}. When we are looking for a witness for a query number $x_i$ we traverse down the tree in logarithmic time, until we reach a leaf with at most $O(X/R)$ 1s. A naive $O(X/R)$ algorithm at the sub-vectors of a leaf will complete the process. Thus the total query time is $O(n^2X/R^2)$.

\bigskip
\noindent
{\bf Partial Convolution via Partial Matrix Multiplications\ }
We make use of \emph{matrix multiplication} in order to efficiently compute the convolutions at the leaves as follows.

We construct a set of matrices $U=\{U_1,U'_1,U_2,U'_2,...,U_{cn/X},U'_{cn/X}\}$, two for each $u_i$ as follows.
The first row in matrix $U_i$ is the sub-vector $u_i$. The $k$th row of $U_i$ is the sub-vector $u_i$ shifted right by $k$, i.e. discard the $k-1$ least significant bits and add $k-1$ $0$s to be the new most significant bits.
Each such matrix is called a \emph{shift right matrix}. In a similar manner, the first row in matrix $U'_i$ is the sub-vector $u_i$. The $k$th row of $U_i$ is the sub-vector $u_i$ shifted left by $k$, i.e. discard the $k-1$ most significant bits and add $k-1$ $0$s to be the new least significant bits. Each such matrix is called a \emph{shift left matrix}

Now, construct a matrix $V$ from the $O(n/X)$ sub-vectors of $v$ such that the $i$th column of $V$ is $v_i$. This matrix has $X$ rows and $O(n/X)$ columns. We partition  $V$ into $O(n/X^2)$ squared matrices $V_1,V_2,...,V_{cn/X^2}$ each having exactly $X$ rows and columns.

The matrix multiplication of $U_i$ and $V_j$ together with the matrix multiplication of $U'_i$ and $V_j$ computes all the shifts of $u_i$ with the $X$ sub-vectors represented by $V_j$. In other words, we are computing the convolution of $X$ leaves with two matrix multiplication of two matrices of size $X\times X$. Hence, the overall time to compute the convolutions of the leaves will be  $O(n/X \cdot n/X^2 \cdot T(X))$ where $T(X)$ is the time needed for multiplying two squared matrices of size $X \times X$. The computation of the convolution of inner vertices can be computed from the convolution of its children. It is straightforward to see that the time of the computations of the convolutions is dominated by the computation of the convolutions in the leaves.

\subsection{Hardness of Partial Matrix Multiplication}

We are now ready to prove the following \CLB{s} on partial matrix multiplication and its indexing variant.

\begin{theorem}~\label{partialMatrix1}
Assume the \ThreeSUM{} conjecture is true.
Then there is no algorithm for the partial matrix multiplication problem running in $O(N^{2-\Omega(1)})$ expected time, even if $|S|$ and the number of ones in the input matrices is $O(N^{2-\Omega(1)})$.
\end{theorem}
\begin{proof}
Consider the reduction from Lemma~\ref{lem-convolution}. Finding $O(n^2/R)$ witnesses from the $O(R^2)$ instances of witnesses convolution problem can be done by using the tree we have just constructed and seeking witnesses within it. Recall that for each instance $(u,v,w)$ of the witnesses convolution problem we only look for witnesses for $O(n/R)$ locations in $w$ (that corresponds to the indices in the bucket $B_{a+b \mod R}$ it represents, see details in Lemma~\ref{lem-convolution}).

As explained in the previous subsection, the total time to construct the quad tree is $O(n/X \cdot n/X^2 \cdot T(X))$ where $T(X)$ is the time needed for multiplying two squared matrices of size $X \times X$. The results of the matrix multiplications of some $V_i$ and all the matrices in $U$ correspond to the convolution of $X$ sub-vectors of $v$ with $u$. However, as noted before, we are interested in computing the result of each convolution only at $O(n/R)$ locations (that correspond to the indices specified by elements of bucket $B_{a+b \mod R}$).
Therefore, after partitioning $u$ and creating the matrices in $U$ which are of size $X \times X$ the number of locations in the result of the multiplication of some $V_i$ with some matrix in $U$ that we are interested in is $X^2/R$ in expectation.
But we will require the worst-case number of locations to be at most $X^2/R$, so for each result that has more than $X^2/R$ locations which we are interested in we will split the computation into several iterations, each time considering a set of different $X^2/R$ locations.
Since the total number of locations over all the results is no more than $n/R$, this only imposes a constant factor time overhead.

Now, in the construction time of the quad tree, instead of the time $T(x)$ for multiplying two squared matrices of size $X \times X$ we just require the time to compute their multiplication in some specific $X^2/R$ locations. Moreover, the number of $1$s in any $V_i$ is $O(X^2/R)$ which is also the number of $1$s in any matrix
in $U$. If the time for computing the multiplication of these sparse matrices at the $O(X^2/R)$ specified locations is $O(X^{2-\Omega(1)})$ then the total construction time will be $O(n^2/X^{1+\Omega(1)})$. Therefore, the total time for constructing all the trees for all $R^2$ instances will be $O(R^2n^2/X^{1+\Omega(1)})$, and the query time is $O(n^2X/R^2)$. Choosing $X=n^{\epsilon}$ for some small constant $0<\epsilon$ and $R=X^{1/2+1/4\epsilon}$ we obtain a truly subquadratic solution for \ThreeSUM{}.
\end{proof}

\begin{theorem}~\label{partialMatrix2}
Assume the \ThreeSUM{} conjecture is true.
Then there is no algorithm for the partial matrix multiplication indexing problem with $O(SIZE(S))$ preprocessing time and $O(N^{2-\Omega(1)})$ query time.
\end{theorem}
\begin{proof}
We make use of Lemma~\ref{lem-convolution} and our quad tree.
This time, we preprocess $V_i$ with a collection $S = \{S_1,S_2,...,S_{O(n/X)}\}$. Each set $S_i \in S$ corresponds to the $O(X^2/R)$ indices of interest in the result of the multiplication of $V_i$ and some matrix in $U$ (see the details in the proof of the previous theorem). After the preprocessing phase we answer queries to compute the partial multiplication of $V_i$ with matrices in $U$ using the indices from sets in $S$. The construction time of the quad tree is $O(n/X \cdot [P(X,n/R)+ n/X^2 \cdot Q(X,n/R)])$, where $P(X,n/R)$ is the preprocessing time of partial matrix multiplication indexing for sparse matrices of size $X \times X$ while $SIZE(S)=O(n/R)$, and $Q(X,n/R)$ is the corresponding query time. Therefore, if there is an algorithm for the partial matrix multiplication indexing problem with $O(SIZE(S)) = O(n/R)$ preprocessing time and $O(X^{2-\Omega(1)})$ query time, then the total construction time for all trees is $O(R^2 \cdot n/X \cdot [n/R + n/X^2 \cdot X^{2-\Omega(1)}]) = O(Rn^2/X+R^2n^2/X^{1+\Omega(1)})$, and the total time spent on all queries is $O(n^2X/R^2)$. Choosing $X=n^{\epsilon}$ for some small constant $0< \epsilon$ and $R=X^{1/2+1/4\epsilon}$ we obtain a truly subquadratic solution to \ThreeSUM{}.
\end{proof}

Notice that Theorems~\ref{partialMatrix1} and \ref{partialMatrix2} hold even when considering the simple case of \emph{boolean} matrix multiplication.

%% file: jumbled_indexing.tex
\section{Hardness of Data Structures for Histogram Indexing}

In order to prove a \CLB{} for both the histogram indexing problem and the histogram (indexing) reporting problem, we will first focus on reducing \ThreeSUM{} to the histogram reporting problem, and then turn our focus to reducing the the histogram reporting problem to the histogram indexing problem.

\subsection{Reducing \ConvolutionThreeSUM{} to Histogram Reporting}\label{ss:C3S_to_HR}

We are given an ordered set $A$ of integers $x_1,x_2,...,x_n$ for which we want to solve \DiffConvolutionThreeSUM{}.
Our methodology here is to encode the input integers into a compact string $S$ so that histogram indexing with carefully chosen query patterns implies a solution to \DiffConvolutionThreeSUM{}.
Since the size of the universe of the input integers can be as large as $n^3$, we hash down the universe size while (almost) maintaining the linearity property of the input.
To do this, we make use of an almost-linear almost-balanced hash function $h:U \rightarrow [R]$ as defined in Section~\ref{sec:preliminaries}, and apply $h$ to all of the input integers.%As usual, for simplicity of presentation we assume $h$ is indeed linear and point out that using an almost-linear hash imposes only constant factor overheads in our complexities.

After utilizing $h$ to compress the input range, we are ready to encode the input and create the string $S$. To do this, we encode each $h(x_k)$ separately, and then concatenate the encodings in the same order as their corresponding original integers in $A$. We use the following encoding scheme, using an alphabet $\Sigma = \{\sigma_0,\sigma_1,,...,\sigma_{\ell-1}\}$. Some other encoding schemes, which surprisingly provide the same bounds, are discussed in Appendix~\ref{sec:app}.

\textbf{Encoding 1}. The encoding will consist of two separate partial encodings concatenated together. The first partial encoding is partitioned into $\ell$ parts which together will represent $h(x_k)$ in base $R^{1/\ell}$.
For $0\leq j \leq \ell -1$ the $j$th part of this first partial encoding is a unary representation of $ p_{j,h(x_k)} = \lfloor h(x_k)/R^{j/\ell} \rfloor \modulo R^{1/\ell}$ using $\sigma_j$, and is denoted by $enc(j,h(x_k)) = \sigma_j^{p(j,h(x_k))}$.
The first partial encoding of $h(x_k)$, which we also call a \emph{regular encoding} of $h(x_k)$, is

 $enc_\ell(h(x_k))  = enc(0,h(x_k)) enc(1,h(x_k)) \cdots enc(\ell-1,h(x_k))
\\  = \sigma_0^{p_{0,h(x_k)}}\sigma_1^{p_{1,h(x_k)}}\cdots \sigma_{\ell-1}^{p_{{\ell-1},h(x_k)}}$.

For the second partial encoding we encode the \emph{complement} of each $enc(j,h(x_k))$ which is the unary representation of $\bar p_{j,h(x_k)} = R^{1/\ell}-(\lfloor h(x_k)/R^{j/\ell} \rfloor \modulo R^{1/\ell})$ using $\sigma_j$, and is denoted by $\overline{enc}(j,h(x_k))$.
The second partial encoding of $h(x_k)$, which we also call a \emph{complement encoding} of $h(x_k)$, is
%\begin{align}
\\ $\overline{enc_\ell}(h(x_k)) =  \overline{enc}(0,h(x_k)) \overline{enc}(1,h(x_k)) \cdots \overline{enc}(\ell-1,h(x_k))
 = \sigma_0^{\bar p_{0,h(x_k)}}\sigma_1^{\bar p_{1,h(x_k)}}\cdots \sigma_{\ell-1}^{\bar p_{{\ell-1},h(x_k)}}$.
%\end{align}
The full encoding of $h(x_k)$ is the concatenation of $\overline{enc_\ell}(h(x_k))$ and $enc_\ell(h(x_k))$ which we denote by $ENC_\ell(h(x_k))$.
Finally, the string $S$ is set to be \\ $ENC_\ell(h(x_1)) ENC_\ell(h(x_2)) \cdots ENC_\ell(h(x_n))$. The size of $S$ is clearly $N=O(\ell\cdot R^{\frac{1}{\ell}}n)$. We denote the substring of $S$ starting at the location of the beginning of $enc_\ell(h(x_i))$ and ending at the location of the end of $\overline{enc_\ell}(h(x_j))$ by $S_{i,j}$.

\smallskip

Consider a Parikh vector $v_k$ obtained from $x_k$ and $h$ where the $r$th element has a count of $\bar p_{r,h(x_k)}+ R^{1/\ell} \cdot (k-1)$. We say that $v_k$ \emph{represents} $x_k$. For a vector $w=(w_0,w_1,...,w_m)$ we define $w^{>>1}=(0,w_0,w_1,...,w_{m-1})$. We also define the \emph{carry set} of $v_k$ to be $V_k=\{v_k+R^{1/\ell}u-u^{>>1}\ |\  u=(u_0,u_1,...,u_{\ell-2},0),\ u_i \in \{0,1\}\ 0 \leq i < \ell-1 \}$. It is easy to see that $|V_k|=2^{\ell-1}$ and that $V_k$ can be obtained from $v_k$ in $O(\ell\cdot 2^{\ell-1})$ time. We call $v_k$ the \emph{base} of $V_k$. We have the following lemma regarding $V_k$:

\begin{mylemma}\label{correctness_lemma}
If there exists a pair $x_i, x_j$ such that $x_k=x_j-x_i$ and $k=j-i$, then the Parikh vector of $S_{i,j}$ must be in $V_k$.
\end{mylemma}

\begin{proof}
Since $h$ is linear we know that $h(x_k)=h(x_j)-h(x_i)$. This is equivalent to saying that $R+R^{\frac{\ell - 1}{\ell}}-h(x_k)=R+R^{\frac{\ell - 1}{\ell}}-[h(x_j)-h(x_i)] = (R+R^{\frac{\ell - 1}{\ell}}-h(x_j))+h(x_i)$. In $S_{i,j}$ we have the full encoding of all integers $x_{i+1},..., x_{j-1}$. There are exactly $k-1$ integers between $x_i$ and $x_j$. Therefore, each of them adds $R^{1/\ell}$ occurrences of each $\sigma_r$ ($0 \leq r \leq l-1$) to $S_{i,j}$. In addition to the full encodings of these integers we have two more partial encodings: $enc_\ell(h(x_i))$ and $\overline{enc_\ell}(h(x_j))$. Notice that $enc_\ell(h(x_i))$ and $\overline{enc_\ell}(h(x_j))$ represent $h(x_i)$ and $R+R^{\frac{\ell - 1}{\ell}}-h(x_j)$, respectively, in base $R^{1/\ell}$. If we look at the vector $v_k$ (the base of $V_k$) after subtracting $(k-1)R^{1/\ell}$ from the count of each character, we obtain the representation of $R+R^{\frac{\ell - 1}{\ell}}-h(x_k)$ in base $R^{1/\ell}$, which intuitively implies that $v_k$ is the Parikh vector that we are looking for. However, it is  possible to generate a carry at each of the $\ell$ digits of the base $R^{1/\ell}$ during the addition of $(R+R^{\frac{\ell - 1}{\ell}}-h(x_j))+h(x_i)$. To handle these carries we consider all possible $2^\ell$ carry scenarios and generate a vector for each of the $2^{\ell-1}$ scenarios. These carry scenarios are exactly represented by the vectors in $V_k$, as each vector $u$ in the definition of $V_k$ specifies the indices in which we have a carry. Hence, the Parikh vector of $S_{i,j}$ must be one of the vectors in $V_k$.
\end{proof}

Thus, we preprocess $S$ with an algorithm for histogram reporting, and then query the resulting data structure with all the vectors in $V_k$, whose base $v_k$ represents some $x_k$, in an attempt to decide if $x_k$ is part of a solution to \DiffConvolutionThreeSUM{}. The reported locations are classified into two types:

\noindent\textbf{Candidates}: Locations where the histogram match begins and ends exactly between the complement and regular encodings of two input integers. All these locations correspond to $x_i$ and $x_j$ such that for the particular $h(x_k)$ for which the query was constructed, we have $h(x_k)=h(x_j)-h(x_i)$ and also $k = j - i$.

\noindent\textbf{Encoding Errors}: All matches that are not candidates.

While encoding errors clearly do not provide a solution for \DiffConvolutionThreeSUM{} on $A$, candidates may also not be suitable for a solution since the function $h$ introduces false-positives.
The following lemma bounds the total expected number of false-positives (both from false-positive candidates and encoding errors) that can be reported by a single query vector (and the vectors in the carry set that it serves as it base). Its proof appears in the appendix.

\begin{mylemma}\label{lem:histogram_fp}
The expected number of false positives that are reported when considering all vectors in $V_k$ (whose base represents $x_k$) as queries is $O(2^{\ell-1} N/R^{1-\frac{1}{\ell}})$.
\end{mylemma}

\begin{proof}
We focus on $v \in V_k$ that is queried when considering $x_k$. This vector $v$ implies the value of $m$ which is the length of substrings of $S$ that can have $v$ as their Parikh vector.
Clearly, there are at most $N$ such substrings. We focus on the substring from location $\alpha$ to location $\alpha+m-1$ in $S$.
Due to our encoding scheme, this substring contains a (possibly empty) suffix of $ENC_\ell(h(x_i))$, for some $x_i$, followed by $k-1$ full encodings of some integers from $A$, and then a (possibly empty) prefix of $ENC_\ell(h(x_j))$ , for some integers $x_i$ and $x_j$. The only way in which we may falsely report location $\alpha$ as a match is if for each $\sigma \in \Sigma$ the number of $\sigma$ characters in the substring of $S$, denoted by $f(\sigma, \alpha, m)$, is equal to the count of $\sigma$ in $v$, denoted by $v_\sigma$. For a given $\sigma$, since the substring contains $k-1$ complete encodings, we can consider $v_\sigma - (k-1) R^{1/\ell} $ which is a function of $\bar p_{r,h(x_k)}$, compared to $f(\sigma, \alpha, m) - (k-1) R^{1/\ell}$. Now, since $\bar p_{r,h(x_k)}$ is uniformly random (due to $h$) in the range $[R^{1/\ell}]$, the probability that they are equal is $R^{-1/\ell}$. This is true for every character $\sigma$ on its own, but when considering all of the $\ell$ characters, once we set the count for the first $\ell-1$ characters the count for the last character completely depends on the other counts. Therefore, the probability that the comparison passes for all of the characters only depends on the first $\ell-1$ characters, and is $1/R^{1-1/\ell}$. By linearity of expectation over all possible locations in $S$ and all $2^{\ell-1}$ vectors in $V_k$, the expected number of false positives is $O(2^{\ell-1} N/R^{1-\frac{1}{\ell}})$.
\end{proof}

\subsection{Hardness of Histogram Reporting}

Utilizing the reduction we have described in the previous section, that transforms an ordered set $A$ to a string $S$, we can prove the following \CLB{}.

\begin{theorem}\label{thm:histogram_indexing_clb}
Assume the \ThreeSUM{} conjecture holds.
The histogram reporting problem for an $N$-length string and constant alphabet size $\ell \geq 2$ cannot be solved using $O(N^{2-\frac{2\gamma}{\ell+\gamma}-\Omega(1)})$ preprocessing time, $O(N^{1-\frac{\gamma}{\ell+\gamma}-\Omega(1)})$ query time and $O(N^{\frac{\gamma \ell}{\ell+\gamma} - \frac{2\gamma}{\ell+\gamma}-\Omega(1)})$ reporting time per item, for any $0 < \gamma < \ell$.
\end{theorem}

\begin{proof}
We follow the reduction in Section~\ref{ss:C3S_to_HR}.
For an instance of the histogram reporting problem on a string of length $N$ denote the preprocessing time by $O(N^\alpha)$, the query time by $O(N^\beta)$ and the reporting time per item by $O(N^\delta)$. The total expected running time used by our reduction to solve \DiffConvolutionThreeSUM{} is $O(N^\alpha)+n \cdot O(N^\beta) + E_{fp}\cdot O(N^\delta)$, where $E_{fp}$ is the expected total number of false positives. This running time must be $\Omega(n^{2-\Omega(1)})$, unless \ThreeSUM{} conjecture is false.

Since $N=O(\ell \cdot R^{\frac{1}{\ell}}n)$ and $E_{fp}=O(n2^\ell N/R^{1-\frac{1}{\ell}})$, then either $(\ell \cdot R^{\frac{1}{\ell}}n)^{\alpha} = \Omega(n^{2-o(1)})$, $(\ell \cdot R^{\frac{1}{\ell}}n)^{\beta} = \Omega(n^{1-o(1)})$, or $n2^\ell (\ell \cdot R^{\frac{1}{\ell}}n)/R^{1-\frac{1}{\ell}} \cdot (\ell \cdot R^{\frac{1}{\ell}}n)^{\delta} = \Omega(n^{2-o(1)})$. Set $R$ to be $n^\gamma$. By straightforward calculations following our choice of $R$ we get that $\alpha = 2 - \frac{2\gamma}{\ell+\gamma} - \Omega(1)$, $ \beta = 1 - \frac{\gamma}{\ell+\gamma} - \Omega(1)$, and $ \delta= \frac{\gamma \ell}{\ell+\gamma} - \frac{2\gamma}{\ell+\gamma}-\Omega(1)$.
\end{proof}

\subsection{From Reporting to Decision: Hardness of Histogram Indexing}

We make use of Theorem~\ref{thm:histogram_indexing_clb} to obtain a \CLB{} on the decision variant of the problem. Amir et al. \cite{ACLL14} proved similar lower bounds based on a stronger \ThreeSUM{} conjecture. Our proof here shows that this stronger assumption is not needed and that the common \ThreeSUM{} conjecture suffices. The idea of the proof is to make the expected number of false-positives small by a suitable choice of $R$.

 \begin{mylemma}\label{lem:histogram_indexing}
Assume the \ThreeSUM{} conjecture holds.
The histogram indexing problem for a string of length $N$ and constant alphabet size $\ell\geq 3$ cannot be solved with $O(N^{2-\frac{2}{\ell-1}-\Omega(1)})$ preprocessing time and $O(N^{1-\frac{1}{\ell-1}-\Omega(1)})$ query time.
\end{mylemma}

\begin{proof}
We follow the reduction in Section~\ref{ss:C3S_to_HR}.
In order to use histogram indexing we will reduce the probability of a false positive for any query to be less than $1/2$.
From Lemma~\ref{lem:histogram_fp} we know that the expected number of false positives due to query is at most $O(\frac{2^{\ell-1}(\ell R^{\frac{1}{\ell}}n)}{R^{1-\frac{1}{\ell}}})$.
By setting $R$ to be $c_1n^{\frac{\ell}{\ell-2}}$ for sufficiently large constant $c_1$ the number of false positives is strictly smaller than $1/2$, which implies immediately that the probability of a false positive is strictly smaller than $1/2$. Therefore, if we were to solve histogram indexing instead of histogram reporting on the same input as in Theorem~\ref{thm:histogram_indexing_clb}, the probability of a false positive is less than $1/2$.
We can make this probability smaller by repeating the process $O(\log n)$ times, each time using a different hash function $h$. This way, the probability that all of the queries that are due to a specific $x_k$ return false positives is less than $1/poly(n)$. If a given $x_k$ passes all of the query processes (that is, a positive answer is received by each one of them), then we can verify that there is indeed a match with this $x_k$ in $O(n)$ time, which will add a negligible cost to the expected running time in the case it is indeed a false positive. Thus, the total expected running time of this procedure is $O(\log{n}(P(N,\ell)+nQ(N,\ell)))$, where $P(N,\ell)$ is the preprocessing time (for input string of length $N$ and alphabet size $\ell$) and $Q(N,\ell)$ is the query time (for the same parameters).
Therefore, unless the \ThreeSUM{} conjecture is false, there is no solution for histogram indexing such that $P(N,\ell)=O(n^{2-\Omega(1)})$ and $Q(N,\ell)=O(n^{1-\Omega(1)})$.
If we plug-in the value of $R$ we have chosen and follow the calculations in the proof of Theorem~\ref{thm:histogram_indexing_clb} (with $\gamma = \frac{\ell}{\ell-2}$), then we obtain that there is no solution for the histogram indexing problem with $P(N,\ell)=O(N^{2-\frac{2}{\ell-1}-\Omega(1)})$ and $Q(N,\ell)=O(N^{1-\frac{1}{\ell-1}-\Omega(1)})$.
\end{proof}

We generalize this \CLB{} by presenting a full-tradeoff between preprocessing and query time. The idea of the proof is to artificially split the encoded string $S$ to smaller parts, so we can have many false positives in $S$, but the probability for a false positive in each part will be small.

\begin{theorem}\label{thm:generalized_histogram_indexing}
Assume the \ThreeSUM{} conjecture holds.
The histogram indexing problem for a string of length $N$ and constant alphabet size $\ell \geq 3$ cannot be solved with $O(N^{2-\frac{2(1-\alpha)}{\ell-1-\alpha}-\Omega(1)})$ preprocessing time and $O(N^{1-\frac{1+\alpha(\ell-3)}{\ell-1-\alpha}-\Omega(1)})$ query time, for any $0 \leq \alpha \leq 1$.
\end{theorem}

\begin{proof}
The idea to get a full-tradeoff between preprocessing time and query time is to artificially split the encoded string $S$ to smaller parts, so we can have many false positives in $S$ but the probability for a false positive in each part will be small. We achieve this artificial split by using special character $\sigma^{*}$ as a separator and construct a search data structure for queries in the following way.

\textbf{Search Structure Construction.} Say we have encoded $A$ using encoding 1 with $\ell-1$ characters and formed a string $S$. Let $v=(v_{{\sigma}_1},v_{{\sigma}_2},...,v_{{\sigma}_{\ell-1}})$ be a query Parikh vector such that the sum of all counts in it is $N/2$. That is, $v$ can be the Parikh vector of a substring of $S$ that has length $N/2$. We add $n^{\alpha}-1$ special characters $\sigma^{*}$ to $S$ for some $\alpha \in [0,1]$. The special characters are added to $S$ at positions $N/2+aN/(2n^{\alpha})$  for $1 \leq a < n^{\alpha}$. We denote the resulting string by $S'$.

Denote by $u_{\sigma}$ the count of character $\sigma$ in the Parikh vector $u$. We generate $n^\alpha$ vectors from our query vector $v$. These vectors are given by the set \\ $V=\{v_i=(v_{{\sigma}_1},v_{{\sigma}_2},...,v_{{\sigma}_{\ell-1}},v_{\sigma^{*}}) |  v_{\sigma^{*}} = i \ 0 \leq i < n^{\alpha} \}$. Querying $S'$ with a vector $v' \in V$ such that $v'_{\sigma^{*}} = i$ can find a substring of $S'$ with matching Parikh vector only if this substring in $S$ (without the special characters $\sigma^{*}$) begins at location $\beta \in [(i-1)N/(2n^{\alpha}),iN/(2n^{\alpha}-1)]$. This is because of the limitations posed by the special character $\sigma^{*}$.

To handle query vectors whose sum of counts is in the interval $[N/4+1,N/2]$, we will add more special characters to $S$. Specifically, we add 1 special character at positions $N/4+iN/(4n^{\alpha})$  for $1 \leq i \leq n^{\alpha}$, 2 special characters at positions $N/2+iN/(4n^{\alpha})$  for $1 \leq i \leq n^{\alpha}$, and 3 special characters at positions $3N/4+iN/(4n^{\alpha})$  for $1 \leq i < n^{\alpha}$. Again, we denote the resulting string $S'$. With the special characters in $S'$, we can handle a query Parikh vector $v$ that the sum of its counts is $len \in [N/4+1,N/2]$ in the following way. We first extend $v$ to contain a count for $\sigma^{*}$ and initialize it to the number of occurrences of $\sigma^{*}$ in the substring $S'[0,len-1]$ (the substring of $S'$ that starts at location $0$ and ends at location $len-1$). We call this vector $v^0$. We create sequence of $O(n^{\alpha})$ query vectors $v^i$ such that $v^i$ is created from $v^0$ by adding $i$ to $v^{0}_{\sigma^{*}}$. It is straightforward to observe that each query vector $v^i$ can match at most $N/(4n^{\alpha})$ locations in $S'$. That is, by adding just $O(n^\alpha)$ special symbols to $S$ we artificially split the string to $O(n^\alpha)$ parts.

\emph{Multi-level Construction.} We explain how to handle a query Parikh vector \\ $v=(v_{{\sigma}_1},v_{{\sigma}_2},...,v_{{\sigma}_{\ell-1}})$ that the sum of all counts in it is $len \leq |S|/4$. To do this, we create a structure with $O(\log n)$ levels, such that each one of them is created by cutting $S$ and adding $O(n^\alpha)$ special characters. Level 0 will be the string we constructed previously. For level $i$, we cut the string $S$ to $2^i$ parts of (almost) equal length by cutting at positions $j|S|/2^i$ of $S$ for $1 \leq j \leq 2^i-1$. We adjust the cutting so that each part will end at the end of a full encoding of some number $x_j$. We denote the resulting parts by $S_1^i=\{S_{1,1}^i,S_{1,2}^i,...,S_{1,2^i}^i\}$. Moreover, we also cut $S$ at positions $|S|/2^{i+1} + j|S|/2^i$ for $0 \leq j \leq 2^i-1$ . Again, we adjust the cutting so that each part will end at the end of a full encoding of some number in $A$. We denote the resulting parts except the first and last one by $S_2^i=\{S_{2,1}^i,S_{2,2}^i,...,S_{2,2^i-1}^i\}$. Observe that $|S_1^i \bigcup S_2^i| = O(2^i)$, the length of each part in $S_1^i \bigcup S_2^i$ is $O(N/2^i)$, and the total length of all parts is $O(N)$. We stop the process at level $\alpha \log n$ in which we have $O(n^\alpha)$ parts in $S_1 \bigcup S_2$. We add to each part $S_{j,k}^i$ $O(n^{\alpha}/2^i)$ special characters $\sigma^{*}$. We place them similarly to how we have done it in level 0. Specifically, we add 1 special character at positions $N/2^{i+2}+aN/(4n^{\alpha})$  for $1 \leq a \leq n^{\alpha}/2^i$, 2 special characters at positions $N/2^{i+1}+aN/(4n^{\alpha})$  for $1 \leq a \leq n^{\alpha}/2^i$, and 3 special characters at positions $3N/2^{i+2}+aN/(4n^{\alpha})$ for $1 \leq a < n^{\alpha}/2^i$. After creating all these levels and parts we preprocess each of these parts as an instance of histogram indexing. Then, given a a query Parikh vector $v=(v_{{\sigma}_1},v_{{\sigma}_2},...,v_{{\sigma}_{\ell-1}})$ that the sum of its counts is $len$, we first find the right level to handle it. This is the level $i$ for which we have $N/2^{i+2} \leq len < N/2^{i+1}$. if $len \geq N/2$ we use level 0 and if $len < N/n^{\alpha}$ we use the last level. Let the right level be level $i$. We first extend $v$ to contain a count for $\sigma^{*}$ and initialize it to the number of occurrences of $\sigma^{*}$ in the substring $S_{j,k}^i[0,len-1]$ (for some $j$ and $k$). We call this vector $v^0$. We create a sequence of $O(n^{\alpha}/2^i)$ query vectors $v^i$ such that $v^i$ is created from $v^0$ by adding $i$ to $v^{0}_{\sigma^{*}}$. We now query each part in $S_1^i \bigcup S_2^i$ with each vector from the set of vectors we have created. It is straightforward to observe that we have $O(n^{\alpha})$ queries in total and each query can match at most $N/(4n^{\alpha})$ locations of some part in $S_1^i \bigcup S_2^i$.

\medskip
\textbf{Time analysis.} Now, we analyse the preprocessing time and query time of our construction. Say we use encoding 1 to create $S$. The expected number of false-positives for a query without the special character is $O(N/R^{1-\frac{1}{\ell-1}})$ (see Lemma~\ref{lem:histogram_fp}). In Lemma~\ref{lem:histogram_indexing} we required that this number will be strictly less than 1/2 in order to guarantee that with high probability we have no false positives for a given query (after amplification). Following our construction, we have $O(n^\alpha)$ queries (which is smaller than $cn^\alpha$ for some constant $c>0$) for each query vector $v \in V_k$, so that each of these queries covers only some of the substrings of $S$. Therefore, we can allow the number of false positives per vector $v$ to be strictly less than $1/2cn^\alpha$, which will be strictly less than 1/2 for each of the $O(n^\alpha)$ queries corresponding to $v$. We can reduce the probability for a false positive per query by creating $O(\log n)$ copies of our structure, each time using a different hash function $h$. Then, for every query vector $v$ we query all $O(\log n)$ structures. This way, the probability that all of the queries return a false positive is less than $1/poly(n)$. This is due to the fact that using different $h$ does not change the order of the encoded numbers in the string $S$. It only changes their encoding. Therefore, a query vector can match the exact same locations in all $O(\log n)$ structures.

That being said, we only need to ensure that $(\ell R^{1/(\ell-1)}n)/R^{1-\frac{1}{\ell-1}}<1/2cn^\alpha$. This is satisfied if we choose $R=c'n^{\frac{(\ell-1)(1-\alpha)}{\ell-3}}$ for some sufficiently large constant $c'$. By this choice of $R$ we get that $N=O(n^{\frac{\ell-2-\alpha}{\ell-3}})$. Our structure has $O(\log n)$ levels. The total length of the string parts in each level is $O(N)$. Moreover, in each level we allow $O(n^{\alpha})$ false positives in total. Therefore, the construction time is dominated by the preprocessing time of level 0, as all levels have the same length and the preprocessing time is $O(N^\beta)$ for some $\beta \geq 1$. We construct $O(\log n)$ copies of our structure, so the overall preprocessing time is $\tilde{O}(P(N,\ell))$ which is the preprocessing time for histogram indexing on input string of length $N$ and alphabet size $\ell$. We have $n$ numbers in $A$ and for each number $x_k$ in $A$ we have $2^{\ell-2}$ queries (vectors in $V_k$). These are the $O(2^{\ell}n)$ queries we have in the analysis in proof of Lemma~\ref{lem:histogram_indexing}.  Using our special character we have created $O(n^\alpha)$ query vectors for each of these $O(2^{\ell}n)$ queries. Therefore, the total number of query vectors is ,therefore, $O(2^{\ell}n^{1+\alpha})$. Let $Q(N,\ell)$ be the query time for histogram indexing on input string of length $N$ and alphabet size $\ell$. The total expected running time of our method for solving \ThreeSUM{} is $\tilde{O}((P(N,\ell)+n^{1+\alpha}Q(N,\ell)))$ for constant alphabet size $\ell$. Therefore, unless the \ThreeSUM{} conjecture is false, there is no solution for histogram indexing such that $P(N,\ell)=O(n^{2-\Omega(1)})$ and $Q(N,\ell)=O(n^{1-\alpha-\Omega(1)})$.
If we plug-in the value of $R$ we have chosen, then we obtain that there is no solution for the histogram indexing problem with $P(N,\ell)=O(N^{\frac{2(\ell-3)}{\ell-2-\alpha}-\Omega(1)})$ and $Q(N,\ell)=O(N^{\frac{(\ell-3)(1-\alpha)}{\ell-2-\alpha}-\Omega(1)})$.

\emph{Improving the tradeoff.} We get full tradeoff between preprocessing time and query time for each alphabet size $\ell$. We wish that this tradeoff will contain the points given by Lemma~\ref{lem:histogram_indexing} when $\alpha$ gets close to zero. However, we have "wasted" one character for $\sigma^{*}$. We can bypass this problem by encoding $S$ using encoding 3 (see Appendix~\ref{sec:app}). In this encoding we have a special character $\sigma_{\ell-1}$ that separates between partial encodings. We can use this character also as $\sigma^{*}$. In order to do so, we need to ensure that there will be no confusion between the two uses of $\sigma_{\ell-1}$. We call the characters $\sigma_{\ell-1}$ used for separating between partial encodings \emph{separating} characters and the characters $\sigma_{\ell-1}$ used as $\sigma^{*}$ \emph{splitting} characters. For a query vector $v$ we have $v_{\sigma_{\ell-1}}=v_{\sigma_{\ell-1}}^{sep}+v_{\sigma_{\ell-1}}^{spl}$, where $v_{\sigma_{\ell-1}}^{sep}$ is the number of separating characters counted by $v_{\sigma_{\ell-1}}$ and $v_{\sigma_{\ell-1}}^{spl}$ is the number of splitting characters counted by $v_{\sigma_{\ell-1}}$. Recall that in our construction each part of $S$ in each level has regions with 1, 2 or 3 special characters $\sigma^{*}$. Instead of this amount of special characters, we place $c^{*}$, $2c^{*}$ or $3c^{*}$ special characters for some constant $c^{*}>2$. Now, for each query vector $v$ we create the query vectors $v_i \in V$ by adding each time $c^{*}$ to the count of $\sigma_{\ell-1}$ (that represents $\sigma^{*}$). For a query vector $v$ the count of all the characters except $\sigma_{\ell-1}$ implies the number of full encoding we have. Let $v$ be a query vector for $x_k$ then the number of full encodings will be $k-1$. Therefore, if $v$ matches a substring of some part of $S$, the separating characters in this substring must be at least $v_{\sigma_{\ell-1}}^{sep}-2$. Out of $v_{\sigma_{\ell-1}}^{sep}$ the only two separating characters that can be mistakenly considered as splitting characters are those between the partial encodings of some $x_i$ and $x_j$ such that $k=j-i$. However, as splitting characters come in groups of $c^{*}$ characters for $c^{*}>2$, we cannot mistakenly exchange two separating characters for a complete group of splitting characters. Therefore, the only possibility for a confusion between the two types is by $v$ matching a substring that contains a partial group of splitting characters. This can happen only at the suffix or prefix of the substring. Nevertheless, we can prevent this possibility by adjusting the placement of the splitting characters to the end of the full encoding in which they occur. That is, if a group of splitting characters occurs inside the full encoding of some $x_i$ we move them to the end of its encoding in $S$.

Consequently, by the use of $\sigma_{\ell-1}$ as both separating and splitting character we can get the same splitting effect without the need for an additional character. If we redo the running time calculations, but this time with $S$ encoded using $\ell$ characters, we get that unless the \ThreeSUM{} conjecture is false, there is no solution for histogram indexing such that $P(N,\ell)=O(N^{\frac{2(\ell-2)}{\ell-1-\alpha}-\Omega(1)})$ and $Q(N,\ell)=O(N^{\frac{(\ell-2)(1-\alpha)}{\ell-1-\alpha}-\Omega(1)})$.

For encoding 3 we need at least 4 characters in the alphabet (at least two characters for encoding a number, one for padding and another to separate between partial encodings) to extend the result to $\ell=3$ we use encoding 2 (see Appendix~\ref{sec:app}). In this encoding we have  $\sigma_{\ell-1}$ as a padding character. Specifically, for each number $x_i$ we add $\sigma_{\ell-1}$ $(\ell +1)R^{\frac{1}{\ell-1}}$ times as padding. We can use $\sigma_{\ell-1}$ not just for padding, but also as our special character $\sigma^{*}$. Consequently, by choosing $c^{*}>2(\ell+1)R^{\frac{1}{\ell-1}}$ and following the same analysis as before we prove that the same \CLB{} holds also for alphabet of size $3$ (note that the number of $\sigma_{\ell-1}$ characters we have added only increase the size of the string by some constant factor).
\end{proof}

%% file: summary_of_applications.tex
\section {Appendix: Summary of Applications}\label{sec:summary_of_applications}

\begin{center}
\begin{table}[!htbp]
    \begin{tabular}{ l l l l l l }
    \hline

    Problem & Type & \parbox[t]{1.8cm}{Preprocessing\\ Time}  & \parbox[t]{1.4cm}{Query\\ Time} & \parbox[t]{1.6cm}{Reporting\\ Time} & Remarks \\
    \noalign{\vskip 2mm}
    \hline \hline
     \parbox[t]{1.8cm}{Convolution\\ Witnesses} & \parbox[t]{1.1cm}{Reporting} & $\Omega(N^{2-\alpha})$ & $\Omega(N^{1-\alpha/2})$ & $\Omega(N^{\alpha/2-o(1)})$ & \parbox[t]{2.2cm}{$[$Theorem~\ref{thm-convolutiuon}$]$ \\ $0 < \alpha < 1$ } \\
     \noalign{\vskip 2mm}
    \toprule
    \parbox[t]{1.8cm}{Partial\\ Convolution} & \parbox[t]{1.1cm}{Decision} & $\Omega(N^{1-o(1)})$ & --- &     --- & \parbox[t]{2.2cm}{$[$Theorem~\ref{thm:partial-convolution}$]$\\ Sparse input: \\ $\#1 < N^{1-\Omega(1)}$; \\Sparse required\\ output: \\ $|S|<N^{1-\Omega(1)}$}\\
    \hline
    \parbox[t]{1.8cm}{Partial\\ Convolution\\ Indexing} & \parbox[t]{1.1cm}{Decision} & $\Omega(N^{2-o(1)})$ & $\Omega(N^{1-o(1)}) $ & --- &  \parbox[t]{2.2cm}{$[$Theorem~\ref{thm:partial-convolution-indexing}$]$\\ Sparse input: \\ $\#1 < N^{1-\Omega(1)}$; \\Sparse required\\ output: \\ $|S|<N^{1-\Omega(1)}$}\\
    \hline
    \parbox[t]{1.8cm}{Partial\\ Matrix\\ Multiplication} & \parbox[t]{1.1cm}{Decision} & $\Omega(N^{2-o(1)})$ & ---  & --- & \parbox[t]{2.2cm}{$[$Theorem~\ref{partialMatrix1}$]$\\ Sparse input: \\ $\#1 < N^{2-\Omega(1)}$; \\ Sparse required\\ output:\\ $|S|<N^{2-\Omega(1)}$}\\
    \hline
    \parbox[t]{1.8cm}{Partial\\ Matrix\\ Multiplication\\ Indexing}& \parbox[t]{1.1cm}{Decision} & $\Omega(SIZE(S))$ & $\Omega(N^{2-o(1)})$  & --- & \parbox[t]{2.2cm}{$[$Theorem~\ref{partialMatrix2}$]$\\ Sparse input: \\ $\#1 < N^{2-\Omega(1)}$; \\Sparse required\\ output: \\ $|S_i|<N^{2-\Omega(1)}$; \\ $SIZE(S)=$\\ $\sum_{i=1}^{k} |S_i|$}\\
    \noalign{\vskip 2mm}
    \hline \hline
    \noalign{\vskip 1mm}
    \parbox[t]{1.8cm}{Histogram\\ Reporting}& \parbox[t]{1.1cm}{Reporting} & \parbox[t]{2.6cm}{$\Omega(N^{2-\frac{2\gamma}{\ell+\gamma}-o(1)})$} & $\Omega(N^{1-\frac{\gamma}{\ell+\gamma}-o(1)})$  & \parbox[t]{2.35cm}{$\Omega(N^{\frac{\gamma \ell}{\ell+\gamma} - \frac{2\gamma}{\ell+\gamma}-o(1)})$} & \parbox[t]{2.2cm}{$[$Theorem~\ref{thm:histogram_indexing_clb}$]$\\ alphabet size:\\ $\ell \geq 2$;\\ $0 < \gamma < \ell$}\\
    \toprule
    \noalign{\vskip 1mm}
    \parbox[t]{1.8cm}{Histogram\\ Indexing}& \parbox[t]{1.1cm}{Decision} & \parbox[t]{2.6cm}{$\Omega(N^{2-\frac{2(1-\alpha)}{\ell-1-\alpha}-o(1)})$} & \parbox[t]{2.7cm}{$\Omega(N^{1-\frac{1+\alpha(\ell-3)}{\ell-1-\alpha}-o(1)})$} & --- & \parbox[t]{2.2cm}{$[$Theorem~\ref{thm:generalized_histogram_indexing}$]$\\ alphabet size:\\ $\ell > 2$;\\ $0 \leq \alpha \leq 1$}\\
    \hline
    \end{tabular}
    \caption{Summary of \CLB{s} proved in this paper. In this table $N$ is the size of vectors, strings and the dimension of matrices. $\#1$ refers to the number of ones in the input. The rows in this table are interpreted to mean that there is no data structure that beats these preprocessing, query, and reporting (if exists) complexities at the same time. For partial convolution and matrix multiplication the \CLB{} on the preprocessing time should be interpreted as a \CLB{} on the total running time as these are offline problems.}
    \label{table:summary_of_applications}
    \end{table}
\end{center}

\newpage

\begin{center}
\begin{table}[h]
    \begin{tabular}{ l l l l }
    \hline

     & Amir et al.~\cite{ACLL14} & This Paper & Remarks\\
    \hline \hline
    \parbox[t]{5cm}{\CLB{} for preprocessing time} & $\Omega(N^{2-\frac{2}{\ell-1}-o(1)})$ & $\Omega(N^{2-\frac{2(1-\alpha)}{\ell-1-\alpha}-o(1)})$ & \parbox[t]{2.2cm}{alphabet size:\\ $\ell>2$;\\ $0 \leq \alpha \leq 1$}\\
    \hline
    \parbox[t]{5cm}{\CLB{} for query time} & $\Omega(N^{1-\frac{1}{\ell-1}-o(1)})$ & $\Omega(N^{1-\frac{1+\alpha(\ell-3)}{\ell-1-\alpha}-o(1)})$ & \parbox[t]{2.2cm}{alphabet size:\\ $\ell>2$;\\ $0 \leq \alpha \leq 1$}\\
    \hline
    \parbox[t]{5cm}{Hardness assumption} & \parbox[t]{2.5cm}{Strong 3SUM\\ conjecture} & \parbox[t]{2.5cm}{Standard 3SUM\\ conjecture} & \\
    \hline
    \parbox[t]{5cm}{\CLB{s} with full tradeoff between preprocessing and query time for every alphabet size.} & No & Yes & \\
    \hline
    \parbox[t]{5cm}{Separation between binary alphabet and alphabet of size>2 for subquadratic preprocessing time and $\tilde{O}(1)$ query time.} & No & Yes & \\
    \hline
    \parbox[t]{5cm}{\CLB{s} for the reporting version of histogram indexing.} & No & Yes & \\
    \hline
    \hline
    \end{tabular}
    \caption{Comparison of the results in~\cite{ACLL14} and this paper regarding histogram indexing. In this table $N$ is the size of the input string. In each column the first two rows in this table are interpreted to mean that there is no data structure that beats these preprocessing and query complexities at the same time.}
    \label{table:histogram_indexing_comparison}
    \end{table}
\end{center} 

%% file: appendix.tex
\section {Appendix: Other Encodings for Histogram Indexing}\label{sec:app}

We describe two additional encodings of an instance of \DiffConvolutionThreeSUM{} to a string, which then can be used to solve the instance using histogram reporting and histogram indexing. Surprisingly, although these encodings have somewhat different flavour they all give the same \CLB{s} as encoding 1 given in Section~\ref{ss:C3S_to_HR}.

\textbf{Encoding 2}. The idea behind this encoding is to use encoding 1 on alphabet of size $\ell -1$ (instead of $\ell$), and then use the additional character as a padding tool so that the lengths of the partial encodings of each integer will be exactly $\ell R^{1/(\ell-1)}$.
For each $h(x_i)$ we first create both $enc_{\ell-1}(h(x_i))$ and $\overline{enc_{\ell-1}}(h(x_i))$ as defined in encoding 1. We denote by $|enc|$ the length of some encoding $enc$. We first encode $h(x_i)$ by $enc2_{\ell}(h(x_i))=(\sigma_{\ell-1})^{\ell R^{1/(\ell-1)}-|enc_{\ell-1}(h(x_i))|} \cdot enc_{\ell-1}(h(x_i))$. Then, we encode it by $\overline{enc2_{\ell}}(h(x_i)) = \sigma_{\ell-1}^{\ell R^{1/(l-1)}-|\overline{enc_{\ell-1}}(h(x_i))|} \cdot \overline{enc_{\ell-1}}(h(x_i))$. Finally, the full encoding, denoted by $ENC2(h(x_i))$, is the concatenation of the two encodings.

Using this encoding we create a string $S$ by concatenating $ENC2(h(x_k))$ for all $1 \leq k \leq n$. We also construct query pattern for each $h(x_k)$, which is a Parikh vector $v_k$ such that the $r$th element of this vector is $\bar p_{r,h(x_k)} + R^{1/{(\ell-1)}} \cdot (k-1)$ for $0 \leq r \leq \ell-2$. For $r=\ell-1$ the value of the $r$th element will be $\ell R^{1/(\ell-1)}-|\overline{enc_{\ell-1}}(h(x_k))|+\ell R^{1/(\ell-1)}-|enc_{\ell-1}(h(x_k))| + (\ell+1) R^{1/{(\ell-1)}} \cdot (k-1)$. We can also define a carry set $V_k$ in an analog way to the carry set defined in Section~\ref{ss:C3S_to_HR}. This set contains all $2^{\ell-2}$ vectors obtained from $v_k$ in $O(\ell 2^{\ell})$ time in order to handle carry issues. It can be shown in a similar manner to Lemma~\ref{correctness_lemma} that if there exist a pair $x_i, x_j$ such that $x_k=x_j-x_i$ and $k=j-i$, then $S_{i,j}$ had some $v \in V_k$ as its Parikh vector.

Regarding the length of the encoded string $S$ and the expected number of false-positives induced by the reduction we have the following lemma:

\begin{mylemma}
The string $S$ created by the reduction from \DiffConvolutionThreeSUM{} to histogram reporting using encoding 2 is of size $N=O(\ell R^{\frac{1}{\ell-1}}n)$. The number of expected false positives when considering a query vector $v \in V_k$ is $O(N/R)$.
\end{mylemma}

\begin{proof}
Each $x_k$ is embedded by the hash function $h$ to a number in $[R]$. Then, it is divided to $\ell-1$ parts (that represent $h(x_k)$ in base $R^{\frac{1}{\ell-1}}$) each of them encoded by unary encoding. Therefore, each of the encoded parts is of size at most $R^{\frac{1}{\ell-1}}$, while the last letter of the alphabet used for padding completes it exactly to $R^{\frac{1}{\ell-1}}$. The same calculation is true to both $enc2_{\ell}(h(x_i))$ and $\overline{enc2_{\ell}}(h(x_i))$. That being said, the total length of the string using encoding 2 is $N=O(\ell R^{\frac{1}{\ell-1}}n)$.

Now, we consider the expected number of false positives for some query vector $v \in V_k$. In a similar manner to what we have shown for encoding 1 (see Lemma~\ref{lem:histogram_fp}), we are interested in the probability $\Pr[\bigwedge_{j=0}^{l-2} (\bar p_{j,h(x_k)}=f(\sigma_j,i,m,k))]$, where $f(\sigma_j,i,m,k)=f(\sigma_j, i, m) - (k-1) R^{1/(\ell-1)}$ ($f(\sigma_j,i,m)$ represents the number of occurrences of $\sigma_j$ in a substring of $S$ of length $m$ starting at location $i$. $f(\sigma_j,i,m,k)$ represents the same number without taking into account the occurrences of $\sigma_j$ in the $k-1$ fully encoded integers in that substring).  The properties of the (almost) linear hash function guarantee that the value of $h(x_k)$ is expected to be uniformly distributed in its range. Therefore, we have $\Pr[\bar p_{j,h(x_k)}=f(\sigma_j,i,m,k)] = 1/R^{\frac{1}{\ell-1}}$. We conclude that $\Pr[\bigwedge_{j=0}^{l-2} (\bar p_{j,h(x_k)}=f(\sigma_j,i,m,k))]\approx 1/R$ as the number of occurrences of all letters (except the one used for padding) is (almost) independently and uniformly distributed. By this, we have that the expected number of false-positives for a query vector $v \in V_k$ is $O(N/R)$ .
\end{proof}

We can plug $N$ and the expected number of false-positives, which we have calculated in the previous lemma, in the proof of Theorem~\ref{thm:histogram_indexing_clb}. Choosing $R=n^{\gamma '}$ and finally substituting $\gamma ' = \gamma (\ell-1)/\ell$ we obtain the same \CLB{s} that appear in Theorem~\ref{thm:histogram_indexing_clb}. Moreover, the same \CLB{s} as in Lemma~\ref{lem:histogram_indexing} can be achieved using this encoding, if we choose $R$ to be $c_2n^{\frac{l}{\ell-1}}$ for sufficiently large constant $c_2$ so the expected number of false positives we have is strictly less than $1/2$.

\textbf{Encoding 3}. In this encoding we first create the same string $S$ as in encoding 2 using $\ell-1$ characters (instead of $\ell$, and then add special characters to the string. We use $\sigma_{\ell-1}$ as the special character (note that this time  $\sigma_{\ell-1}$ was not used by encoding 2, as it make use only of the first $\ell-1$ characters in the alphabet). We add $\sigma_{\ell-1}$ at every position $j \ell R^{1/\ell-1}$ in the string $S$ for $0 \leq j \leq 2n$. That is, we use $\sigma_{\ell-1}$ to mark the border of the partial encodings of each number $x_i$ in the string created by encoding 2.

We construct a query for each $h(x_k)$, which is a Parikh vector $v_k$ such that the $r$th ($0 \leq r \leq \ell-2$) element of this vector is $\bar p_{r,h(x_k)} + R^{1/{(\ell - 2)}} \cdot (k-1)$. For $r=\ell-2$ the value of the $r$th element will be $(\ell-1) R^{1/(\ell-2)}-|\overline{enc_{l-2}}(h(x_k))|+(\ell-1) R^{1/(\ell-2)}-|enc_{\ell-2}(h(x_k))|+ \ell R^{1/{(\ell-2)}} \cdot (k-1)$. Finally, for $r=\ell-1$ the value of the $r$th element will be $2k+1$. We can also define a carry set $V_k$ in an analog way to the carry set defined in Section~\ref{ss:C3S_to_HR}. This set contains all $2^{\ell-3}$ vectors obtained from $v_k$ in $O(\ell 2^{\ell})$ time for handling carry issues. As before, It can be shown that if there exist a pair $x_i, x_j$ such that $x_k=x_j-x_i$ and $k=j-i$, then $S_{i,j}$ had some $v \in V_k$ as its Parikh vector.

Regarding the length of the encoded string $S$ and the expected number of false-positives induced by the reduction we have the following lemma:

\begin{mylemma}
The string $S$ created by the reduction from \DiffConvolutionThreeSUM{} to histogram reporting using encoding 3 is of size $N=O(\ell R^{\frac{1}{\ell-2}}n)$. The number of expected false positives when considering a query vector $v \in V_k$ is $O(n/R)$.
\end{mylemma}

\begin{proof}
Each $x_i$ is embedded by the hash function $h$ to a number in $[R]$. Then, it is divided to $\ell-2$ parts (that represent $h(x_k)$ in base $R^{\frac{1}{\ell-2}}$) each of them encoded by unary encoding. Therefore, each of the encoded parts is of size at most $R^{\frac{1}{\ell -2}}$, while $\sigma_{\ell- 2}$, that is used for padding, completes each of them exactly to $R^{\frac{1}{\ell-2}}$. The character $\sigma_{\ell-1}$ occurs only $O(n)$ times. That being said, the total length of the string using encoding 3 is of size $N=O(\ell R^{\frac{1}{\ell -2}}n)$.

Now, we consider the expected number of false positives for some query vector $v \in V_k$. The use of the padding character $\sigma_{\ell- 2}$ and the special character $\sigma_{\ell-1}$ guarantees that every reported match starts exactly between $\overline{enc2_{\ell}}(h(x_i))$ and $enc2_{\ell}(h(x_i))$ and ends exactly between $\overline{enc2_{\ell}}(h(x_j))$ and $enc2_{\ell}(h(x_j))$ for some $i$ and $j$ such that $k=j-i$. Therefore, no false-positives are created by encodings errors. That being said, the expected number of false-positives for a query vector $v \in V_k$ is exactly the number of false-positives introduced by the use of the hash function which is $O(n/R)$.
\end{proof}

We can plug $N$ and the expected number of false-positives, which we have calculated in the previous lemma, in the proof of Theorem~\ref{thm:histogram_indexing_clb}. Choosing $R=n^{\gamma ''}$ and finally substituting $\gamma '' = \gamma (\ell-2)/\ell$ we obtain the same \CLB{s} that appear in Theorem~\ref{thm:histogram_indexing_clb}. Moreover, the same \CLB{s} as in Lemma~\ref{lem:histogram_indexing} can be achieved using this encoding, if we choose $R$ to be $c_3n$ for sufficiently large constant $c_3$ so the expected number of false positives we have is strictly less than $1/2$.

We note that the encoding 2 requires one extra padding character and encoding 3 requires another special character, so for them Theorem~\ref{thm:histogram_indexing_clb} and Lemma~\ref{lem:histogram_indexing} hold from slightly larger $\ell$.